\documentclass[sigconf]{aamas} 
\usepackage{balance} % for balancing columns on the final page
\usepackage{amsmath}
\usepackage{amsthm}
\usepackage{mathrsfs}
\usepackage{color}
\usepackage{mathtools}
\usepackage{bm}
\usepackage{bbm}
\usepackage{natbib}
\usepackage{hyperref}
\usepackage{algorithm}
\usepackage{algpseudocode}
\usepackage{sgame}
\usepackage[capitalise]{cleveref}
\usepackage{subcaption}
\usepackage{apxproof}

\DeclareMathOperator*{\argmax}{arg\,max}

\newtheorem{theorem}{Theorem}[section]

\newtheorem*{theorem4}{Theorem 3.4}
\newtheorem*{theorem5}{Theorem 3.5}
\newtheorem*{theorem6}{Theorem 3.6}
\newtheorem{lemma}[theorem]{Lemma}
\newtheorem*{lemma1}{Lemma 3.1}

\theoremstyle{definition}
\newtheorem{definition}{Definition}[section]
\newtheorem{exmp}{Example}
\newtheorem{remark}{Remark}
\newtheorem*{remark1}{Remark 1}

\setcopyright{ifaamas}
\acmConference[AAMAS '23]{Proc.\@ of the 22nd International Conference
on Autonomous Agents and Multiagent Systems (AAMAS 2023)}{May 29 -- June 2, 2023}
{London, United Kingdom}{A.~Ricci, W.~Yeoh, N.~Agmon, B.~An (eds.)}
\copyrightyear{2023}
\acmYear{2023}
\acmDOI{}
\acmPrice{}
\acmISBN{}

\acmSubmissionID{266}

\title[Bridging the Gap Between Single and Multi Objective Games]{Bridging the Gap Between Single and Multi Objective Games}

\author{Willem Röpke}
\affiliation{
  \department{Artificial Intelligence Lab}
  \institution{Vrije Universiteit Brussel}
  \country{Belgium}}
\email{willem.ropke@vub.be}

\author{Carla Groenland}
\affiliation{
  \department{Department of Mathematics}
  \institution{Universiteit Utrecht}
  \country{The Netherlands}}
\email{c.e.groenland@uu.nl}

\author{Roxana R\u{a}dulescu}
\affiliation{
  \department{Artificial Intelligence Lab}
  \institution{Vrije Universiteit Brussel}
  \country{Belgium}}
\email{roxana.radulescu@vub.be}

\author{Ann Now\'{e}}
\affiliation{
  \department{Artificial Intelligence Lab}
  \institution{Vrije Universiteit Brussel}
  \country{Belgium}}
\email{ann.nowe@vub.be}

\author{Diederik M. Roijers}
\affiliation{
  \department{Artificial Intelligence Lab}
  \institution{Vrije Universiteit Brussel}
  \country{Belgium}}
\email{diederik.roijers@vub.be}

\begin{abstract}
A classic model to study strategic decision making in multi-agent systems is the normal-form game. This model can be generalised to allow for an infinite number of pure strategies leading to continuous games. Multi-objective normal-form games are another generalisation that model settings where players receive separate payoffs in more than one objective. We bridge the gap between the two models by providing a theoretical guarantee that a game from one setting can always be transformed to a game in the other. We extend the theoretical results to include guaranteed equivalence of Nash equilibria. The mapping makes it possible to apply algorithms from one field to the other. We demonstrate this by introducing a fictitious play algorithm for multi-objective games and subsequently applying it to two well-known continuous games. We believe the equivalence relation will lend itself to new insights by translating the theoretical guarantees from one formalism to another. Moreover, it may lead to new computational approaches for continuous games when a problem is more naturally solved in the succinct format of multi-objective games.
\end{abstract}
\keywords{Game theory; Continuous game; Multi-objective; Nash equilibria}

\newcommand{\BibTeX}{\rm B\kern-.05em{\sc i\kern-.025em b}\kern-.08em\TeX}
\definecolor{celestialblue}{rgb}{0.29, 0.59, 0.82}

\begin{document}

%%% The following commands remove the headers in your paper. For final 
%%% papers, these will be inserted during the pagination process.

\pagestyle{fancy}
\fancyhead{}

%%% The next command prints the information defined in the preamble.

\maketitle 

%%%%%%%%%%%%%%%%%%%%%%%%%%%%%%%%%%%%%%%%%%%%%%%%%%%%%%%%%%%%%%%%%%%%%%%%

\section{Introduction}
\label{sec:introduction}
Connecting seemingly unrelated models can bridge together research communities and complement missing pieces in either setting. To that extent, we identify two well-known extensions of normal-form games and show that there are underlying equivalences which may be exploited for theoretical and algorithmic contributions. On the one hand, we consider continuous games which allow for an infinite amount of actions for players \citep{stein2008separable}. 
An example of such games can be found in economic models where firms have to set a price over a continuous range which maximises their profit in a competitive environment \citep{judd2012finding}. On the other hand, we consider Multi-Objective Normal-Form Games (MONFGs) \citep{blackwell1954analog}, which extend normal-form games by returning a vectorial payoff rather than a scalar payoff. Multi-objective games can for example be applied for scheduling household appliances in residential buildings with the objectives to minimise electricity cost while adhering as much as possible to the desired operating period \cite{lu2022multiobjective}.

Our theoretical results establish a novel equivalence relation between continuous games and MONFGs. The equivalence is shown to leave the underlying game dynamics intact, including Nash equilibria. The established connection allows one to straightforwardly transfer known results and algorithms from either model to the other, thus bridging the gap between them. From a theoretical perspective, much more is known about continuous games than MONFGs, enabling rapid advances in the latter model. Conversely, due to their succinct format, MONFGs are amenable to general algorithmic solutions which can then be applied to continuous games. The main contributions are summarised as follows\footnote{The supplementary material of this work can be found in \cite{ropke2023bridging}.}:

% We show that any continuous game with convex pure strategy sets can be mapped to an equivalent MONFG and vice versa. Our equivalence notion requires that a bijective function is defined which leaves the underlying game dynamics intact. As the game dynamics are unaffected, we can further guarantee that a pure strategy Nash equilibrium in a continuous game corresponds to a mixed strategy Nash equilibrium in an equivalent MONFG. The established connection allows one to straightforwardly transfer known results and algorithms from either model to the other, which we demonstrate by giving the first general Nash equilibrium existence guarantee in MONFGs. We empirically validate the applicability of our contributions by computing a Nash equilibrium in two continuous games with an extension of the fictitious play algorithm to MONFGs. Our results show that the algorithm reliably converges to a Nash equilibrium in the continuous game, even when the game itself does not satisfy the conditions for equivalence and an approximation is used.

\begin{itemize}
    \item We define pure strategy equivalence between a continuous game and MONFG. This equivalence relation formalises the existence of a bijective function that maps pure strategies from the continuous game to mixed strategies in the MONFG and for which the utilities remain equal.
    \item We show that for every continuous game with convex strategy sets, a pure strategy equivalent MONFG can be constructed. We also show this in the other direction.
    \item We introduce hierarchical strategies for MONFGs and define mixed strategy equivalence between these strategies and the mixed strategies in a continuous game.
    %\item We show that a joint strategy is a pure strategy Nash equilibrium in a continuous game if and only if it is a mixed strategy equilibrium in an equivalent MONFG. We extend this for mixed strategy equilibria in the continuous game and hierarchical strategy equilibria in the MONFG.
    \item We show that a pure strategy Nash equilibrium in a continuous game is a mixed strategy equilibrium in an equivalent MONFG. We extend this for mixed strategy equilibria in the continuous game and hierarchical equilibria in the MONFG.
    \item We guarantee the existence of a hierarchical Nash equilibrium in MONFGs with continuous utility functions. %We give a first example of the usefulness of our connection by deducing that an existence guarantee can be given for a Nash equilibrium in hierarchical strategies in MONFGs with continuous utility functions.
    \item We demonstrate the algorithmic value of our equivalence notion by computing a Nash equilibrium in two continuous games, namely a polynomial game and a Bertrand price game, using a fictitious play algorithm for MONFGs.%an extension of the fictitious play algorithm for MONFGs.
\end{itemize}

%In \cref{sec:background} we introduce the different games we consider in this work and define Nash equilibria in them. \Cref{sec} formally defines our novel equivalence notion, pure strategy equivalence, and shows that a game from one class can always be transformed to one in the other class while keeping Nash equilibria intact. In \cref{sec:empirical-results} we demonstrate our results empirically. Finally, \cref{sec:conclusion} provides 
% I HATE THESE KINDS OF FILLER TEXTS IN PAPERS SO I WOULD LIKE TO AVOID IT IF POSSIBLE.

%%%%%%%%%%%%%%%%%%%%%%%%%%%%%%%%%%%%%%%%%%%%%%%%%%%%%%%%%%%%%%%%%%%%%%%%

\section{Background}
\label{sec:background}
\subsection{Continuous Games}
\label{sec:continuous-game}
Continuous games, sometimes referred to as infinite games, extend the normal-form game model to include settings where players have a nonempty compact metric space of pure strategies, rather than a finite set. Intuitively, this means that players may have an infinite number of actions to choose from. We define a continuous game as follows \citep{stein2008separable},
\begin{definition}[Continuous game]
\label{def:continuous-game}
A continuous game is a tuple $(N, \mathcal{S}, v)$, where: 
\begin{itemize}
    \item $N$ is a finite set of $n$ players, indexed by $i$;
    \item $\mathcal{S} = S_1 \times \dots \times S_n$, where $S_i$ is a nonempty compact metric space of pure strategies available to player $i$;
    \item $v = (v_1, \dots , v_n)$ where $v_i : \mathcal{S} \to \mathbb{R}$ is a real-valued and continuous utility function for player $i$.
\end{itemize}
\end{definition}
The set of mixed strategies for player $i$ in a continuous game is defined as the set of Borel probability measures on $S_i$ and is denoted by $\mathcal{B}(S_i)$. We define the expected utility for player $i$ for a given joint mixed strategy $\mu \in \mathcal{B}(\mathcal{S})$ as follows \citep{stein2008separable,fu2021existence},
\begin{equation}
\label{eq:eu-ms-cg}
    v_i(\mu) = \int_{S_1 \times \cdots \times S_n} v_i(s_1, \cdots , s_n) d\mu_1(s_1) \cdots d\mu_n(s_n).
\end{equation}
We refer readers unfamiliar with these concepts from measure theory to a brief treatment of them in \cref{ap:measure-theory}.

\subsection{Multi-Objective Normal-Form Games}
\label{sec:monfg}
Multi-Objective Normal-Form Games (MONFGs) are a generalisation of (single-objective) normal-form games to vectorial payoffs. This is formalised as follows \citep{radulescu2020utilitybased}.
\begin{definition}[Multi-objective normal-form game]
\label{def:monfg}
A (finite, $n$-player) multi-objective normal-form game is a tuple $(N, \mathcal{A}, \bm{p})$, with $d$ objectives, where: 
\begin{itemize}
    \item $N$ is a finite set of $n$ players, indexed by $i$;
    \item $\mathcal{A} = A_1 \times \dots \times A_n$, where $A_i$ is a nonempty finite set of actions available to player $i$;
    \item $\bm{p} = (\bm{p}_1, \dots , \bm{p}_n)$ where $\bm{p}_i : \mathcal{A} \to \mathbb{R}^d$ is the vectorial payoff function for player $i$.
\end{itemize}
\end{definition}
A mixed strategy for any given player is defined as a probability distribution over their set of actions. A set of probability distributions over a finite number of points is known as a probability simplex, formally defined below.

\begin{definition}
\label{def:prob-simplex}
A probability $k$-simplex $\Delta^{k}$ is a set of points for which,
\begin{equation*}
    \Delta^{k} = \left \{ \left(x_0, \dots, x_k \right) \in \mathbb{R}^{k+1} \Biggm\vert \sum_{i=0}^k x_i = 1 \text{ and } x_i \geq 0 \text{ for } i = 0, \dots, k \right \}.
\end{equation*}
\end{definition}

As such, a mixed strategy for player $i$ is a probability distribution $\delta_i \in \Delta^{k_i}$ where $k_i = |A_i|-1$. Let $\Delta = \Delta^{k_1} \times \dots \times \Delta^{k_n}$ be the set of joint mixed strategies. The expected payoff for player $i$ of a mixed strategy $\delta \in \Delta$ is then naturally defined as,

\begin{equation}
\label{eq:exp-vec-payoff}
   \bm{p}_i\left(\delta\right) = \sum_{a \in \mathcal{A}}\bm{p}_i(a)\prod_{j = 1}^n \delta_j(a_j).
\end{equation}
 
Multi-objective decision making frequently assumes a utility-based approach \cite{roijers2017multiobjective}, which is also a standard game-theoretic approach. In the MONFG model, the utility of a player is based on a particular trade-off of the various payoffs, that is, we make the additional assumption that a utility function  $u_i: \mathbb{R}^d \to \mathbb{R}$ is known for any player $i$. We compute the utility of a mixed strategy $\delta$ as,
\begin{equation}
\label{eq:monfg-u}
   u_i\left(\bm{p}_i\left(\delta\right)\right) = u_i\left(\sum_{a \in \mathcal{A}}\bm{p}_i(a)\prod_{j = 1}^n \delta_j(a_j)\right).
\end{equation}

We note that an alternative definition exists for the utility of mixed strategies in multi-objective games, where players first apply the utility function to the payoffs and subsequently compute their expected utility. In the multi-objective decision making literature, this method is referred to as the expected scalarised returns criterion, while the method from \cref{eq:exp-vec-payoff} is referred to as the scalarised expected returns criterion \citep{hayes2022practical, radulescu2020multiobjective}.

\subsection{Nash Equilibria}
We consider a central solution concept in both models, namely the Nash equilibrium (NE). Informally, a Nash equilibrium is a joint strategy from which no player can unilaterally deviate and improve their utility. \Cref{def:continuous-nash-equilibrium} defines this in continuous games, while \cref{def:monfg-nash-equilibrium} defines this in MONFGs.

\begin{definition}[Nash equilibria in continuous games]
\label{def:continuous-nash-equilibrium}
A mixed strategy profile $\mu^\ast$ is a Nash equilibrium if,
\begin{equation*}
    v_i(\mu^\ast_i, \mu^\ast_{-i}) \geq v_i(\mu_i, \mu^\ast_{-i}),
\end{equation*}
for all players $i$ and mixed strategies $\mu_i \in \mathcal{B}(S_i)$.
\end{definition}

\begin{definition}[Nash equilibria in multi-objective normal-form games]
\label{def:monfg-nash-equilibrium}
A mixed strategy profile $\delta^\ast$ is a Nash equilibrium if,
\begin{equation*}
    u_i(\delta^\ast_i, \delta^\ast_{-i}) \geq u_i(\delta_i, \delta^\ast_{-i}),
\end{equation*}
for all players $i$ and mixed strategies $\delta_i \in \Delta^{k_i}$.
\end{definition}
An important early result is that every continuous game must have a mixed strategy Nash equilibrium \citep{glicksberg1952further}. In MONFGs however, Nash equilibria are not guaranteed to exist under nonlinear utility functions \citep{radulescu2020utilitybased}. In \cref{sec:mixed-strategy-equivalence}, we expand further on this issue and provide a novel Nash equilibrium existence result for MONFGs.

%%%%%%%%%%%%%%%%%%%%%%%%%%%%%%%%%%%%%%%%%%%%%%%%%%%%%%%%%%%%%%%%%%%%%%%%

\section{Equivalence Relation}
\label{sec:equivalence-relation}
Our main contribution establishes an equivalence relation between continuous and multi-objective games. With this goal in mind, we first introduce a special game, called an identity game, which plays a crucial role in bridging the two models. Next, we prove that a bijective mapping for pure strategies in the continuous game to mixed strategies in the MONFG always exists. To complete the full equivalence, we introduce a novel concept for MONFGs, which we call hierarchical strategies. Lastly, we show that Nash equilibria are preserved in the process, allowing us to guarantee a Nash equilibrium in hierarchical strategies in MONFGs.

\subsection{Identity Game}
\label{sec:identity-game}
An identity game returns, as the name suggests, a payoff vector equal to the strategy it received as input. We will use such games in \cref{sec:pure-strategy-equivalence} to prove an equivalent MONFG can be constructed for every continuous game and vice versa.

\begin{lemma}[Identity Game]
\label{lemma:identity-game}
For any finite set of players and finite sets of pure strategies, there exists a set of payoff functions $\bm{p}$ such that for each player $i$, $\bm{p}_i(\delta) = \delta$. 
\end{lemma}

\begin{proofsketch}
The main idea behind the proof is to define the payoff vectors for all pure strategies as the pure strategy itself. Note that, strictly speaking, this relies on the notation of joint strategies, which we define as a vector of concatenated individual strategies. It is then possible to show, using the law of total probability, that this results in the desired property.
\end{proofsketch}

A formal proof is given in \cref{ap:identity-game}. To illustrate the payoff mechanism in more detail, we provide an example below.

\begin{figure}[th]
    \centering
    \begin{game}{2}{2}
                  & $A$        & $B$\\
        $A$   & $(1, 0, 1, 0); (1, 0, 1, 0)$       & $(1, 0, 0, 1); (1, 0, 0, 1)$ \\
        $B$   & $(0, 1, 1, 0); (0, 1, 1, 0)$       & $(0, 1, 0, 1); (0, 1, 0, 1)$
    \end{game}
    \caption{The identity game for a 2-player 2-action setting.}
    \label{fig:identity-game-2-2}
\end{figure}

\begin{exmp}
\label{ex:identity-game}
Consider the identity game in \cref{fig:identity-game-2-2}. Assume that player one plays the mixed strategy $\delta_1 = (\frac{1}{2}, \frac{1}{2})$ and player two plays the mixed strategy $\delta_2 = (\frac{1}{3}, \frac{2}{3})$. This leads to a joint strategy $\delta = (\delta_1, \delta_2) = (\frac{1}{2}, \frac{1}{2}, \frac{1}{3}, \frac{2}{3})$. According to \cref{lemma:identity-game}, the expected payoff vector should then also be $(\frac{1}{2}, \frac{1}{2}, \frac{1}{3}, \frac{2}{3})$. We verify this:

\begin{align*}
    \bm{p}_i(\delta) & = \sum_{a \in \mathcal{A}}\bm{p}_i(a)\prod_{j = 1}^n \delta_j(a_j) \\
    & = (1, 0, 1, 0) \cdot \frac{1}{6} + (1, 0, 0, 1) \cdot \frac{1}{3} + (0, 1, 1, 0) \cdot \frac{1}{6} + (0, 1, 0, 1) \cdot \frac{1}{3} \\
    & = \left(\frac{1}{6}, 0, \frac{1}{6}, 0 \right) + \left(\frac{1}{3}, 0, 0, \frac{1}{3} \right) + \left(0, \frac{1}{6}, \frac{1}{6}, 0 \right) + \left(0, \frac{1}{3}, 0, \frac{1}{3} \right) \\
    & = \left(\frac{1}{2}, \frac{1}{2}, \frac{1}{3}, \frac{2}{3}\right) \\
    & = \delta.
\end{align*}
\end{exmp}

\subsection{Pure Strategy Equivalence}
We introduce a novel equivalence notion between continuous games and MONFGs, called Pure Strategy Equivalence (PSE). Informally, two games are pure strategy equivalent whenever the pure strategies from the continuous game can be bijectively mapped to mixed strategies in the MONFG while keeping the corresponding utilities equal. We formally define this below.

\begin{definition}[Pure strategy equivalence]
\label{def:ps-equivalence}
A continuous game $G_c = (N_c, \mathcal{S}, v)$ is pure strategy equivalent to a finite multi-objective normal-form game $G_m = (N_m, \mathcal{A}, \bm{p})$ with utility functions $u$ when there exists a tuple of functions $(\pi, \varphi)$ such that:
\begin{itemize}
    \item $\pi: N_c \to N_m$ is a bijective function called the \emph{player bijection};
    \item $\varphi = \varphi_1 \times \dots \times \varphi_n$ where $\forall i \in N_c, \varphi_i: S_{i} \to \Delta^{k_{\pi(i)}}$ is a continuous bijective function with a continuous inverse, called the \emph{strategy bijection};
    \item $\forall i \in N_c, \forall s \in \mathcal{S}, v_{i}(s) = u_{\pi(i)}\left(\bm{p}_{\pi(i)}\left(\varphi\left(s\right)\right)\right)$. %\carla{perhaps it is nicer to define $u_{m,\pi(i)}(s)$ elsewhere, to keep the notation uniform and this definition as easy to read as possible}.
\end{itemize}
\end{definition}

We first show that for every MONFG there exists a pure strategy equivalent continuous game. We note that this property has already been used \citep[Lemma~1]{ropke2022nash}, but has not yet been explicitly described in terms of pure strategy equivalence.

\label{sec:pure-strategy-equivalence}
\begin{theorem}
\label{th:monfg-pse-cg}
For every multi-objective normal-form game with continuous utility functions, there exists a pure strategy equivalent continuous game.
\end{theorem}

\begin{proof}
Let $G_m = (N_m, \mathcal{A}, \bm{p})$ be a multi-objective normal-form game and let $u$ be the set of continuous utility functions used by players in $G_m$. We construct a continuous game $G_c = (N_c, \mathcal{S}, v)$ that is pure strategy equivalent to $G_m$.

First, let $N_c = N_m$, making the player bijection $\pi = \mathbbm{1}$. To simplify the notation, we can directly substitute $\pi(i) = i$.

Recall that the set of mixed strategies for player $i$ in $G_m$ was defined as a simplex $\Delta^{k_i}$.
For all $S_i \in \mathcal{S}$, we define $S_{i} = \Delta^{k_i}$, thus satisfying the condition of being a nonempty compact metric space. As such, each strategy bijection $\varphi_i = \mathbbm{1}$.

Finally, define each $v_{i} = u_{i} \circ \bm{p}_i$. Observe that each $v_{i}$ is continuous as it is a composition of two continuous functions. As such, $G_c$ is pure strategy equivalent to $G_m$.
\end{proof}
% \begin{proofsketch}
% Recall that each individual set of mixed strategies in a finite MONFG is a probability simplex. If we define these simplices to be the strategy sets in the continuous game, we can define the utility functions $u_{c,i}$ for the continuous game to simply be the composition of the utility functions $u_{m,i}$ and payoff functions $\bm{p}_i$ from the MONFG. 
% \end{proofsketch}

The proof presented here provides an explicit construction of a pure strategy equivalent continuous game for any MONFG, which we outline in \cref{ap:constructing-games} (see \cref{alg:cg-construction}). Moreover, as the strategy sets in both games are equal, the strategy bijection is simply the identity function. This is critical from an algorithmic viewpoint, as it ensures that the mapping may be computed efficiently. This result has two important implications. First, observe that we may also perform the construction in single-player settings, such as those studied in multi-objective planning and reinforcement learning. This suggests that algorithms designed for continuous action spaces can be used whenever the utility function of the agent is known a priori. Secondly, whenever the resulting utility functions are (twice) differentiable, the resulting game falls under the class of differentiable games, which may be solved using efficient gradient-based methods \cite{letcher2019differentiable}.

We now show a converse to \cref{th:monfg-pse-cg}, namely that every continuous game with convex strategy sets can be mapped to a pure strategy equivalent MONFG. %A formal proof of \cref{th:cg-pse-monfg} is presented in \cref{ap:pure-strategy-equivalence}.

\begin{theorem}
\label{th:cg-pse-monfg}
For every continuous game whose strategy spaces are convex subsets of an Euclidean space, there exists a pure strategy equivalent multi-objective normal-form game.
\end{theorem}

% \begin{proofsketch}
% The main idea behind the proof is to construct an MONFG using the identity game payoffs and reuse the utility functions from the continuous game. As we assumed convexity of the strategy sets, it is known that a homeomorphism between each strategy set and a probability simplex can be constructed. These simplices are the mixed strategies in the MONFG. The utility functions $u_{m,i}$ in the MONFG can then be defined as the composition between the utility functions of the continuous game and the inverse of each homeomorphism.
% \end{proofsketch}

\begin{proof}
Let $G_c = (N_c, \mathcal{S}, v)$ be a continuous game where each $S_i \in \mathcal{S}$ is a compact, convex and nonempty subset of Euclidean space. We construct a finite multi-objective normal-form game $G_m = (N_m, \mathcal{A}, \bm{p})$ with utility functions $u$ that is pure strategy equivalent to $G_c$.

%, i.e. $\forall S_i \in \mathcal{S}: S_i \subseteq \mathbb{R}^{k_i}$.

The player bijection is trivial by letting $N_m = N_c$ and thus $\pi = \mathbbm{1}$. To simplify the notation, we can directly substitute $\pi(i) = i$ and refer to $N_m$ or $N_c$ simply as $N$.

It is a known property that all compact convex subsets of $\mathbb{R}^k$ with a nonempty interior are homeomorphic to the probability $k$-simplex \citep[Theorem~16.4]{bredon1993general}. Let us first assume that each strategy set $S_i$ indeed has a nonempty interior. This means that for every player $i$ there exists a continuous bijective function $f_i$ with a continuous inverse $f^{-1}_i$ such that,
\begin{equation}
    f_i: S_{i} \to \Delta^{k_i}.
\end{equation}
Recall that the set of mixed strategies $\Delta^{k_{i}}$ for player $i$ in $G_m$ is defined as the probability simplex over their actions. Therefore, we may use this homeomorphism to construct the set of mixed strategies over $k_i + 1$ actions,
\begin{equation}
    \forall i \in N: f_i(S_{i}) = \Delta^{k_i}.
\end{equation}
Note that the vertices of the simplex naturally represent the actions in $G_m$. This ensures that each joint strategy can be bijectively mapped from $G_c$ to $G_m$ with
\begin{equation}
    \forall s \in \mathcal{S}: f(s) = f_1(s_{1}) \times \dots \times f_n(s_{n}).
\end{equation}
We can therefore define the strategy bijection $\varphi = f$. Because $f$ is a homeomorphism, it is by definition a continuous bijective function with a continuous inverse.

Let the payoff functions $\bm{p}$ of $G_m$ be the payoffs of the identity game from Lemma \ref{lemma:identity-game}. We define the utility functions $u$ in $G_m$ as follows,
\begin{equation}
    \forall i \in N, u_{i} = v_{i} \circ f^{-1}_i.
\end{equation}
We show that this ensures that,
\begin{equation}
    \forall i \in N, \forall s \in \mathcal{S}: v_{i}(s) = u_{i}\left( \bm{p}_{i}\left(\varphi_i\left(s\right)\right)\right).
\end{equation}
Substituting the necessary values, we get $\forall i \in N, \forall s \in \mathcal{S}$, 
\begin{align}
    u_{i}\left( \bm{p}_{i}\left(\varphi_i\left(s\right)\right)\right) & = u_{i}\left(\bm{p}_{i}\left(f_i\left(s\right)\right)\right) \\
    & = u_{i}\left(\bm{p}_{i}\left(\delta\right)\right) \\
    & = u_{i}\left(\delta\right) \\
    & = v_{i} \circ f^{-1}_i \left(\delta\right) \\
    & = v_{i} (s).
\end{align}
As such, $G_c$ is pure strategy equivalent to $G_m$.

In the case where one or more $S_{i}$ has an empty interior, those $S_{i}$ have a nonempty interior with respect to their affine span. Then $S_{i}$ is homeomorphic to the $k$-simplex, where $k = \text{dim aff}(S_i)$. The remainder of the construction follows analogous to before.
\end{proof}

The main idea behind the proof is to construct an MONFG using the identity game payoffs and reuse the utility functions from the continuous game. As we assumed convexity of the strategy sets, it is known that a homeomorphism between each strategy set and a probability simplex can be constructed. These simplices are the mixed strategies in the MONFG. The utility functions in the MONFG can then be defined as the composition between the utility functions of the continuous game and the inverse of each homeomorphism.

The construction defined in this proof also establishes a computational approach to transform a continuous game to an MONFG, formalised in \cref{ap:constructing-games} (see \cref{alg:monfg-construction}). However, contrary to the other direction, the strategy bijection does appear here. As such, the algorithm requires the strategy bijection to be explicitly defined which poses two distinct challenges. First, these strategy bijections may not be easily obtainable, in which case we present a standard approach that initially goes through a unit ball and subsequently to the probability simplex as a suitable first attempt in \cref{ap:constructing-homeomorphisms}. Second, the resulting functions are not guaranteed to be efficiently computable, hence rendering the procedure intractable for some applications. We leave the application of \cref{th:cg-pse-monfg} to such scenarios for future work.

An interesting implication of \cref{th:cg-pse-monfg} is that it can straightforwardly be extended to continuous games with non-convex strategy sets whenever they are still homeomorphic to a simplex. However, while convexity of strategy sets is a frequently made assumption for continuous games, there can be games of interest which do not have this property. In this case, one can approximate the original continuous game with another continuous game that does satisfy the convexity requirement by, for example, taking the convex hull of each strategy set. In a subsequent stage, a pure strategy equivalent MONFG can be constructed for this approximate game, resulting in an MONFG which is approximately PSE to the original continuous game. We demonstrate this approach in \cref{sec:bertrand-price-game} and show that it may still succeed in capturing the original game. %For future work, we aim to further develop this approximation technique. %and provide bounds on the errors which may be introduced this way.

Given the constructions outlined in \cref{th:monfg-pse-cg,th:cg-pse-monfg}, it raises the question of whether they lead to unique descriptions of pure strategy equivalent games. The following remark shows that this is not the case. We show this formally in \cref{ap:constructing-games} by constructing multiple pure strategy equivalent games.
%\carla{Perhaps we could win some space by rewriting this? Bit less nice but we could point to a Remark number in the appendix instead?}
\begin{remark}
\label{re:non-uniqueness}
A continuous game may have multiple pure strategy equivalent multi-objective normal-form games and vice versa.
\end{remark}

An important consequence of \cref{re:non-uniqueness} is that some PSE games can be more amenable to analysis than others. For example, we have shown that for a continuous game a pure strategy equivalent MONFG can be constructed with identity game payoffs. If this continuous game can also be proven to be PSE to an MONFG with a dominated action in its payoff and monotonic utility functions, the latter game will certainly be easier to analyse. Using equivalence between games to solve a difficult game through a simpler equivalent has been studied with success before \citep{heyman2019computation}. 

% Given the constructions in \cref{alg:cg-construction,alg:monfg-construction}, it raises the question of whether they lead to unique descriptions of pure strategy equivalent games. We demonstrate in \cref{ap:constructing-games} that this is not necessarily the case and one can construct multiple pure strategy equivalent games with differing payoff or utility functions. An important consequence of this is that some PSE games can be more amenable to analysis than others. For example, we have shown that for a continuous game a pure strategy equivalent MONFG can be constructed with identity game payoffs. If this continuous game can also be proven to be PSE to an MONFG with a dominated action in its payoff and monotonic utility functions, the latter game will certainly be easier to analyse. Using equivalence between games to solve a difficult game through a simpler equivalent has been studied with success before \citep{heyman2019computation}. 

\subsection{Mixed Strategy Equivalence}
\label{sec:mixed-strategy-equivalence}
Given the results for pure strategy equivalence, a pressing question becomes how to handle the mixed strategies from a continuous game in an equivalent MONFG. For this purpose, we introduce a novel strategy concept in MONFGs, named \emph{hierarchical strategies}. Hierarchical strategies allow for mixing over mixed strategies and are defined for each player $i$ as the set of Borel probability measures on their set of mixed strategies $\Delta^{k_i}$ and denoted by $\mathcal{B}(\Delta^{k_i})$. The expected utility of a hierarchical strategy $\mu$ is then defined analogously to the way expected utility of a mixed strategy in continuous games is defined in \cref{eq:eu-ms-cg}, that is,
\begin{equation}
    u_i(\mu) = \int_{\Delta^{k_1} \times \cdots \times \Delta^{k_n}} u_i\left(\bm{p}_i \left(\delta_1, \cdots , \delta_n\right) \right) d\mu_1(\delta_1) \cdots d\mu_n(\delta_n).
\end{equation}
We stress that the probabilities in a hierarchical strategy cannot be distributed to form an equivalent mixed strategy. This is because the utility of the resulting mixed strategy need not equal the expected utility of the hierarchical strategy, specifically when nonlinear utility functions are used. We illustrate this in \cref{ex:hierarchical-strategy}.

\begin{figure}[h]
    \centering
    \begin{game}{2}{2}
                  & $A$        & $B$\\
        $A$   & $(3, 1); (3, 1)$       & $(1, 3); (1, 3)$ \\
        $B$   & $(1, 3); (1, 3)$       & $(3, 1); (3, 1)$
    \end{game}
    \caption{The game used in \cref{ex:hierarchical-strategy}.}
    \label{fig:hierarchical-strategy}
\end{figure}
\begin{exmp}
\label{ex:hierarchical-strategy}
    Consider the game in \cref{fig:hierarchical-strategy} and utility functions
    \begin{equation}
    \label{eq:balanced}
        u_1(x, y) = u_2(x, y) = x^2 + y^2.
    \end{equation}
    For the row player, we define two strategies, $\delta_{1,1} = (1, 0)$ and $\delta_{1,2} = (0, 1)$. For simplicity, we assume a deterministic strategy $\delta_2 = (1, 0)$ for the column player where they always play action $A$. For both players, the joint strategy $(\delta_{1,1}, \delta_2)$ leads to an expected payoff vector of $(3, 1)$ and $(\delta_{1,2}, \delta_2)$ to an expected payoff vector of $(1, 3)$. Both joint strategies individually result in a utility of 10. 
    
    Consider now the following hierarchical strategy for the row player,
    \begin{equation*}
        \mu_1 = \left(P(\delta_{1,1}) = \frac{1}{2}, P(\delta_{1,2}) = \frac{1}{2}\right).
    \end{equation*}
    This hierarchical strategy denotes the fact that they will play $\delta_{1,1}$ with 50\% probability and $\delta_{1,2}$ with 50\% probability. As both strategies result in a utility of 10, the expected utility of $\mu_1$ is also 10. However, distributing the probabilities in $\mu_1$ to form a mixed strategy, $\delta_{1,3} = (\frac{1}{2}, \frac{1}{2})$, results in an expected payoff vector of $(2, 2)$ with a utility of 8. This demonstrates that a hierarchical strategy cannot be distributed to form an equivalent mixed strategy.
\end{exmp}

While mainly used as an instrument in our analysis, hierarchical strategies may also have an applied use in multi-objective games. Concretely, hierarchical strategies are appropriate to consider when agents have to decide on a strategy which is then executed for a given period without the possibility for downstream adjustments. For example, agents joining an auction may have to commit a single mixed strategy to an automated auctioneer a priori, after which auctions are held for a number of rounds without further input from the agents.

We now define mixed strategy equivalence between a continuous game and an MONFG. Informally, this equivalence notion generalises pure strategy equivalence to relate mixed strategies from the continuous game with hierarchical strategies in the MONFG.

% \begin{definition}[Mixed strategy equivalence]
% \label{def:ms-equivalence}
% A continuous game $G_c$ is mixed strategy equivalent to a finite multi-objective normal-form game $G_m$ when there exists a player bijection and every mixed strategy in $G_c$ can be bijectively mapped to a hierarchical strategy in $G_m$ with equal utility.
% \end{definition}

% \begin{definition}[Mixed strategy equivalence]
% \label{def:ms-equivalence}
% A continuous game $G_c = (N_c, \mathcal{S}, v)$ is mixed strategy equivalent to a finite multi-objective normal-form game $G_m = (N_m, \mathcal{A}, \bm{p})$ with utility functions $u$ when there exists a tuple of functions $(\pi, \psi)$ such that:
% \begin{itemize}
%     \item $\pi: N_c \to N_m$ is a bijective function;
%     \item $\psi = \psi_1 \times \dots \times \psi_n$ where $\forall i \in N_c, \psi_i: \mathcal{B}(S_{i}) \to \mathcal{B}(\Delta_{\pi(i)})$ is a continuous bijective function with a continuous inverse;
%     \item $\forall i \in N_c, \forall \mu \in \mathcal{B}(\mathcal{S}), v_{i}(\mu) = u_{\pi(i)}\left(\psi\left(\mu\right)\right)$.
% \end{itemize}
% \end{definition}

\begin{definition}[Mixed strategy equivalence]
\label{def:ms-equivalence}
Let $G_c = (N_c, \mathcal{S}, v)$ be a continuous game and $G_m = (N_m, \mathcal{A}, \bm{p})$ a finite multi-objective normal-form game with utility functions $u$. $G_c$ is mixed strategy equivalent to $G_m$ if they are pure strategy equivalent with $(\pi, \varphi)$ and there exists a function $\psi$ such that,
\begin{itemize}
    \item $\psi = \psi_1 \times \dots \times \psi_n$ where $\forall i \in N_c, \psi_i: \mathcal{B}(S_{i}) \to \mathcal{B}(\Delta^{k_{\pi(i)}})$ is a bijective function;
    \item $\forall i \in N_c, \forall \mu \in \mathcal{B}(\mathcal{S}), v_{i}(\mu) = u_{\pi(i)}\left(\psi\left(\mu\right)\right)$.
\end{itemize}
\end{definition}

We remark that, by definition, mixed strategy equivalence implies pure strategy equivalence. Moreover, recall that the definition of mixed strategies in continuous games is similar to the definition of the set of hierarchical strategies in MONFGs. This is no coincidence and allows us to show that whenever a continuous game is pure strategy equivalent to an MONFG, it also implies mixed strategy equivalence. A formal proof of this is deferred to \cref{ap:mixed-strategy-equivalence}.
\begin{theorem}
\label{th:cg-mse-monfg}
If a continuous game is pure strategy equivalent to a multi-objective normal-form game, they are also mixed strategy equivalent.
\end{theorem}
\begin{proofsketch}
By definition of pure strategy equivalence, a continuous bijective function with a continuous inverse is given which maps strategies from one game to the other. Moreover, mixed strategies and hierarchical strategies were respectively defined to be the set of all Borel probability measures over their pure strategies and mixed strategies. We can use this fact to show that a measure defined over the strategies in one set must also be a measure in the other set defined over the mapped strategies. This also implies equal utility, thereby completing the proof.
\end{proofsketch}

%Finally, we highlight that this implies that for every MONFG a mixed strategy equivalent continuous game can be constructed. Moreover, it also implies that for every continuous game with convex strategy sets a mixed strategy equivalent MONFG can be constructed. As we have already shown this holds for pure strategy equivalence, \cref{th:cg-mse-monfg} makes this trivial. 

\subsection{Mapping of Nash Equilibria}
\label{sec:mapping-of-ne}
Our final theoretical contributions consider Nash equilibria in pure strategy equivalent games. \Cref{th:ps-ms-ne-equivalence} first specifies that pure strategy Nash equilibria in continuous games correspond to mixed strategy Nash equilibria in pure strategy equivalent MONFGs. Intuitively, this is clear as utilities for pure strategies in the continuous game were already guaranteed to be equal to the utilities for mixed strategies in the MONFG. Therefore, if a joint strategy cannot be improved upon unilaterally in either game, the mapped joint strategy in the related game can also not be improved upon.

\begin{theorem}
\label{th:ps-ms-ne-equivalence}
A pure strategy is a Nash equilibrium in a continuous game if and only if it is a mixed strategy Nash equilibrium in a pure strategy equivalent multi-objective normal-form game.
\end{theorem}

We define a hierarchical Nash equilibrium below, such that we can construct a similar argument to show that mixed strategy NE in continuous games necessarily correspond to a hierarchical strategy NE in pure strategy equivalent MONFGs.

\begin{definition}[Hierarchical Nash equilibria in multi-objective normal-form games]
\label{def:hierarchical-nash-equilibrium}
A hierarchical strategy profile $\mu^\ast$ is a hierarchical Nash equilibrium if,
\begin{equation*}
    u_i(\mu^\ast_i, \mu^\ast_{-i}) \geq u_i(\mu_i, \mu^\ast_{-i}),
\end{equation*}
for all players $i$ and alternative hierarchical strategies $\mu_i \in \mathcal{B}(\Delta^{k_i})$.
\end{definition}

\begin{theorem}
\label{th:ms-hs-ne-equivalence}
A mixed strategy is a Nash equilibrium in a continuous game if and only if it is a hierarchical Nash equilibrium in a pure strategy equivalent multi-objective normal-form game.
\end{theorem}
Complete proofs for \cref{th:ps-ms-ne-equivalence,th:ms-hs-ne-equivalence} are provided in \cref{ap:mapping-ne}. These properties are of significant importance as they introduce algorithmic methods for computing Nash equilibria in either game model to the other. Moreover, as a consequence of \cref{th:ms-hs-ne-equivalence} in particular, we can state the first general result for Nash equilibria to exist in MONFGs if we allow players to assume hierarchical strategies rather than limiting them to mixed strategies only. The proof for \cref{co:hne-guarantee} is given below.
\begin{corollary}
\label{co:hne-guarantee}
Every multi-objective normal-form game with continuous utility functions has a hierarchical Nash equilibrium.
\end{corollary}
\begin{proof}
\Cref{th:monfg-pse-cg} shows that every MONFG with continuous utility functions can be mapped to a pure strategy equivalent continuous game. Furthermore, it is known that a mixed strategy Nash equilibrium exists in every continuous game \citep{glicksberg1952further}. \Cref{th:ms-hs-ne-equivalence} guarantees that a mixed strategy Nash equilibrium in a continuous game is a hierarchical Nash equilibrium in the MONFG, therefore guaranteeing a hierarchical Nash equilibrium in every MONFG.
\end{proof}

%%%%%%%%%%%%%%%%%%%%%%%%%%%%%%%%%%%%%%%%%%%%%%%%%%%%%%%%%%%%%%%%%%%%%%%%

\section{Empirical Results}
\label{sec:empirical-results} 
We provide empirical evidence for the provided theorems and show that it can also be applied to compute approximate equilibria when the strategy sets do not satisfy the necessary conditions for pure strategy equivalence. We adapt the well-known fictitious play algorithm from single-objective games to multi-objective games and use it to compute pure strategy Nash equilibria in continuous games. The results empirically demonstrate the applicability of our contributions and may serve as a useful template for future applications. Our implementation is available at \url{https://github.com/wilrop/pure-strategy-equivalence}.

\subsection{Multi-Objective Fictitious Play}
\label{sec:mo-fp}
Fictitious play aims to learn strategies resulting in a Nash equilibrium through repeated plays of the game. While it is not guaranteed to converge in general-sum single-objective games, fictitious play and its extensions are widely used in practice. In \cref{alg:mo-fp}, we show an extension for fictitious play to multi-objective games. For simplicity, we consider a two-player variant but this can trivially be extended to $n$-player games. In each iteration of the algorithm, players calculate the empirical strategy of their opponent based on their history of play and compute a best response to this strategy. Players subsequently sample an action from their new strategy and update their histories.

\begin{algorithm}[t]
\caption{Multi-Objective Fictitious Play}
\label{alg:mo-fp}
    \begin{algorithmic}[1]
    \Require An MONFG $G = (N, \mathcal{A}, \bm{p})$, utility functions $u$ and maximum timestep $T$
    \Ensure A joint strategy $\delta$
    \State $\delta_1 \gets \frac{\bm{1}}{|A_1|}$
    \State $h_1 \gets \bm{0}$ \Comment{The history for player 1}
    \State $\delta_2 \gets \frac{\bm{1}}{|A_2|}$
    \State $h_2 \gets \bm{0}$
    \For{$t \in {1, \cdots, T}$}
    \State $\tilde{\delta}_2 \gets \frac{h_1}{t}$ \Comment{Compute an empirical mixed strategy}
    \State $\tilde{\delta}_1 \gets \frac{h_2}{t}$
    \State $\delta_1 \gets $ \textproc{BestResponse}$(\bm{p}_1, u_1, \tilde{\delta}_2)$ \Comment{Compute a best response}
    \State $\delta_2 \gets $ \textproc{BestResponse}$(\bm{p}_2, u_2, \tilde{\delta}_1)$
    \State $a_1 \gets a \sim \delta_1$ \Comment{Sample an action from the best response}
    \State $a_2 \gets a \sim \delta_2$
    \State $h_{1, a_2} \gets h_{1, a_2} + 1$ \Comment{Update the history}
    \State $h_{2, a_1} \gets h_{2, a_1} + 1$
    \EndFor
    \State \Return $\delta$
    \end{algorithmic}
\end{algorithm}

Recent work has studied an adaptation of fictitious play to continuous games \citep{ganzfried2021algorithm}. In their algorithm, a growing array of past strategies is kept to later compute a best response to, which imposes a significant memory requirement. A key advantage of our approach is that it only requires an array of fixed length, i.e., one entry per action, where a counter is incremented each time an action is played. The empirical mixed strategy of the opponent is then calculated by taking the relative frequency of each action. A limitation of this approach is that it can only learn pure-strategy equilibria from the continuous game. 

%We note that there exist two distinct variants of fictitious play. In the first variant, players update their strategies simultaneously, while with the second players do this in an alternating fashion. It is known that these variants do not need to follow the same trajectories and even have different convergence guarantees \citep{berger2007brown}.

The fictitious play algorithm shown above appears identical to the original fictitious play algorithm. The exception, however, lies in the best response computation steps. In single-objective games, this can be done efficiently by selecting the action with the highest expected returns, i.e.
\begin{equation}
    BR(A_i, s_{-i}, p_i) = \argmax_{a_i \in A_i}p_i(a_i, s_{-i}).
\end{equation}

In multi-objective games, this approach can only be guaranteed to return a correct best response when employing a quasiconvex utility function \citep{ropke2022nash}. In general MONFGs, the best response can be a mixed strategy and thus requires executing an optimisation subroutine to find the strategy generating the maximum utility. As a best response needs to be a global maximum, this requires the use of a global optimisation algorithm. Under specific utility functions or when approximate best responses suffice, a local optimiser could also be used. %a crucial decision is whether the best response should be locally or globally optimal, thereby prescribing the use of a local or global optimisation algorithm. 

\subsection{Polynomial Game}
\label{sec:polynomial-game}
Polynomial games are a subset of continuous games, where utility functions are guaranteed to be polynomial functions of the player strategies \citep{stein2008separable,stein2011correlated}. We demonstrate that such games can also be represented as an MONFG and may be solved without employing any continuous game or polynomial game specific machinery. We cover a simple example as described by Parrilo \cite{parrilo2006polynomial}. 

Consider a zero-sum game where both players select a strategy from the interval $[-1, 1]$. The utility function for player one is defined as,
\begin{equation}
\label{eq:polynomial-utility}
    v_{1}(x, y) = 2xy^2 - x^2 - y,
\end{equation}
with $x$ the strategy selected by player one and $y$ the strategy of player two. As the game is zero-sum, player two's utility is given by $v_{2}(x, y) = -v_{1}(x, y)$. The utility functions used in the game guarantee the existence of a unique Nash equilibrium in pure strategies where $x^\ast = 0.397$ and $y^\ast = 0.630$.

As the strategy sets are line segments, and thus are 1-simplices, a pure strategy equivalent multi-objective game is guaranteed to exist. To complete the transformation from the polynomial game to a multi-objective game, a strategy bijection $\varphi_i: S_i \to \Delta^{k_i}$ is required for each player $i$. The strategy bijection for both players is given by,
\begin{equation}
\label{eq:strat-bij}
    \varphi_i(s_i) = \left(\frac{s_i - s_{i, \min}}{s_{i, \max} - s_{i, \min}}, 1 - \frac{s_i - s_{i, \min}}{s_{i, \max} - s_{i, \min}}\right),
\end{equation}
where $s_{i, \min} = -1$ and $s_{i, \max} = 1$ for both players. The inverse strategy bijection is given by,
\begin{equation}
\label{eq:inv-strat-bij}
    \varphi^{-1}_i(\delta_i) = s_{i, \min} + \delta_{i,0} \cdot \left(s_{i, \max} - s_{i, \min}\right) .
\end{equation}

The final multi-objective game thus has two players with two actions each and corresponding identity game payoffs. Furthermore, the utility functions for both players are $u_{i} = v_{i} \circ \varphi^{-1}_i$. Because the original utility functions $v_{1}$ and $v_{2}$ guarantee a pure strategy Nash equilibrium in the continuous game, fictitious play is well suited to find the mixed strategy equilibrium in the MONFG. 

We execute the fictitious play algorithm for 200 iterations on the constructed multi-objective game and repeat this for 1000 trials. \Cref{fig:polynomial-game} shows the learned strategies over time, with the shaded area denoting the standard variation at that time. We illustrate the Nash equilibrium $(0.397, 0.630)$ with dotted lines. It is clear that our algorithm learns the equilibrium after approximately 100 iterations and is able to keep improving its strategies closer to the exact equilibrium over time.

\begin{figure}[]
    \centering
    \includegraphics[width=0.75\linewidth]{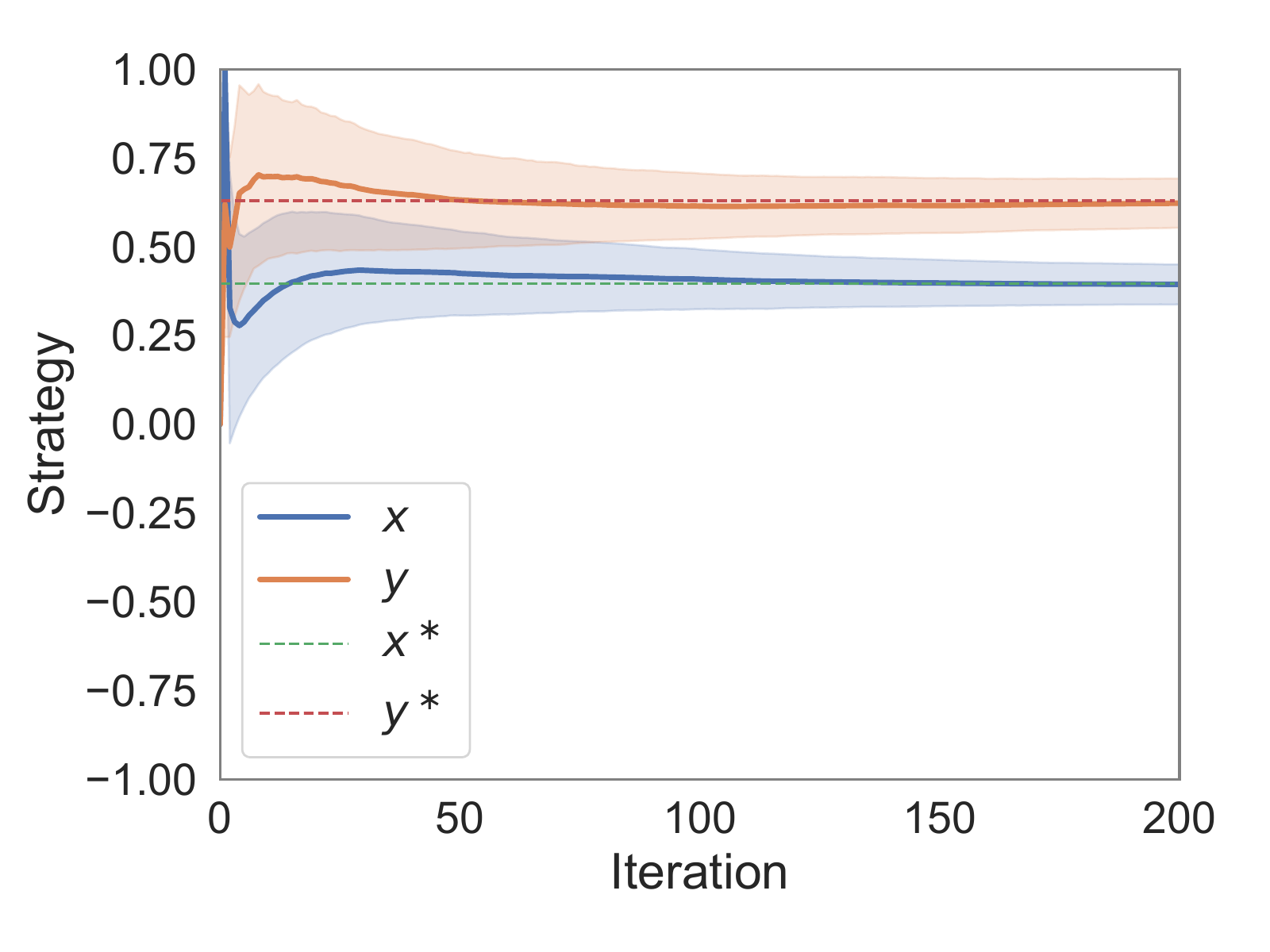}
    \caption{Learning curves for the polynomial game.}
    \label{fig:polynomial-game}
\end{figure}

\subsection{Bertrand Price Game}
\label{sec:bertrand-price-game}
\cref{th:cg-pse-monfg} states that a pure strategy equivalent MONFG is only guaranteed to exist for continuous games whose strategy spaces are convex subsets of an Euclidean space. We demonstrate that pure strategy equivalence can still be applied when this condition is not met by approximating the continuous game. We illustrate this using the Bertrand price game characterised by Judd et al. \cite{judd2012finding}. 

Bertrand price games have been extensively studied as an economic model for determining prices in competitive settings. In this example, we consider two firms, $x$ and $y$, which respectively produce a different good for price $p_x$ and $p_y$. There are three types of customers, which have a distinct demand for both goods. The first type of customer has linear demand curves $d_{x,1}$ and $d_{y,1}$ and only wants the good from firm $x$,
\begin{align}
\label{eq:demand-1}
    d_{x,1}(p_x, p_y) &= a - p_x & d_{y,1}(p_x, p_y) &= 0,
\end{align}
with $a$ signifying all factors, other than price, which influence the demand. The demand function for the third type of customer is defined analogously for the good of firm $y$,
\begin{align}
\label{eq:demand-3}
    d_{x,3}(p_x, p_y) &= 0 & d_{y,3}(p_x, p_y) &= a - p_y.
\end{align}
Finally, the second type of customer has a demand for both goods,
\begin{align}
\label{eq:demand-2}
    d_{x,2}(p_x, p_y) &= n \cdot p_x^{-\sigma}\left(p_x^{1-\sigma} + p_y^{1-\sigma}\right)^{\left(\gamma - \sigma\right)/\left(-1 + \sigma\right)} \\
    d_{y,2}(p_x, p_y) &= n \cdot p_y^{-\sigma}\left(p_y^{1-\sigma} + p_x^{1-\sigma}\right)^{\left(\gamma - \sigma\right)/\left(-1 + \sigma\right)}.
\end{align}
with $n$ the number of type two customers, $\sigma$ the elasticity of substitution between $x$ and $y$ and $\gamma$ the elasticity of demand for the composite good. The total demand for each good, respectively $d_x$ and $d_y$, is given by summing the individual demands for each type. Finally, let $m$ be the unit cost of production for each firm, then the profit for both firms is defined as,
\begin{align}
\label{eq:profit-function}
    r_{x}(p_x, p_y) &= \left(p_x - m\right) \cdot d_x(p_x, p_y) \\
    r_{y}(p_x, p_y) &= \left(p_y - m\right) \cdot d_y(p_x, p_y).
\end{align}

The range of possible prices considered in the game is the open interval $(0, +\infty)$. As such, strategy spaces in the continuous game are non-compact, thus violating a necessary condition for pure strategy equivalence. We can resolve this, however, by making compact convex approximations of the strategy spaces and using these instead. We do this by constraining prices to be in a closed interval $[p_{\min}, p_{\max}]$, which ensures that strategy sets are 1-simplices as in the previous example. Because of this approximation, we may reuse the same strategy bijection as defined in \cref{eq:strat-bij,eq:inv-strat-bij}. Note that approximating the continuous game by altering strategy sets may remove existing equilibria from reach or introduce new ones. In this particular example, as we are both lower and upper bounding the strategy sets, it is possible that an equilibrium falls outside of the bounds and a new equilibrium is created in the MONFG which is not an equilibrium in the original game.

For the following experiments, we define $\sigma = 3, \gamma = 2, n = 2700, m = 1$ and $a = 50$. With these parameters, the price game has two distinct Nash equilibria, shown in \cref{tab:bertrand-eq}. 

\subsubsection{Suitable Approximation}
To give both equilibria a chance of being selected, we set $p_{\min} = 1$ and $p_{\max} = 30$. Every execution of the fictitious play algorithm is run for 200 iterations and results are averaged over 1000 trials as in the previous section.

\begin{table}[t]
\centering
\begin{tabular}{l|l|l|l}
$p_x$  & $p_y$                      & $r_x$                       & $r_y$ \\ \hline
2.168  & 25.157                     & \textbf{724.337}                      & 608.981     \\
25.157 & 2.168                      & 608.981                      & \textbf{724.337}    
\end{tabular}
\caption{The Nash equilibrium strategies and their profits. The highest profit for firm $x$ and $y$ are highlighted.}
\label{tab:bertrand-eq}
\end{table}

In \cref{fig:bertrand-price-game-full}, we show the trajectories leading to the two equilibria. In earlier episodes, the trajectories are non-smooth and show high standard deviation, as the best response computation from a limited history of play leads agents to change strategies rapidly. However, once beliefs converge after approximately 150 episodes, a Nash equilibrium is consistently played. We also find that the learning trajectories for both equilibria are similar, showing that the individual trajectory that is followed is mostly determined by randomisation early on in the learning process. These results demonstrate that multi-objective algorithms can be applied even to approximations of continuous games, given that these approximations sufficiently capture the original game. %We note that due to stochasticity in the underlying global optimiser, a suboptimal response strategy was computed in a small number of trials. However, when enough episodes are left, we find that \cref{alg:mo-fp} can recover and converge to a Nash equilibrium.

\begin{figure}[t]
    \centering
    \begin{subfigure}[b]{0.45\linewidth}
        \includegraphics[width=\linewidth]{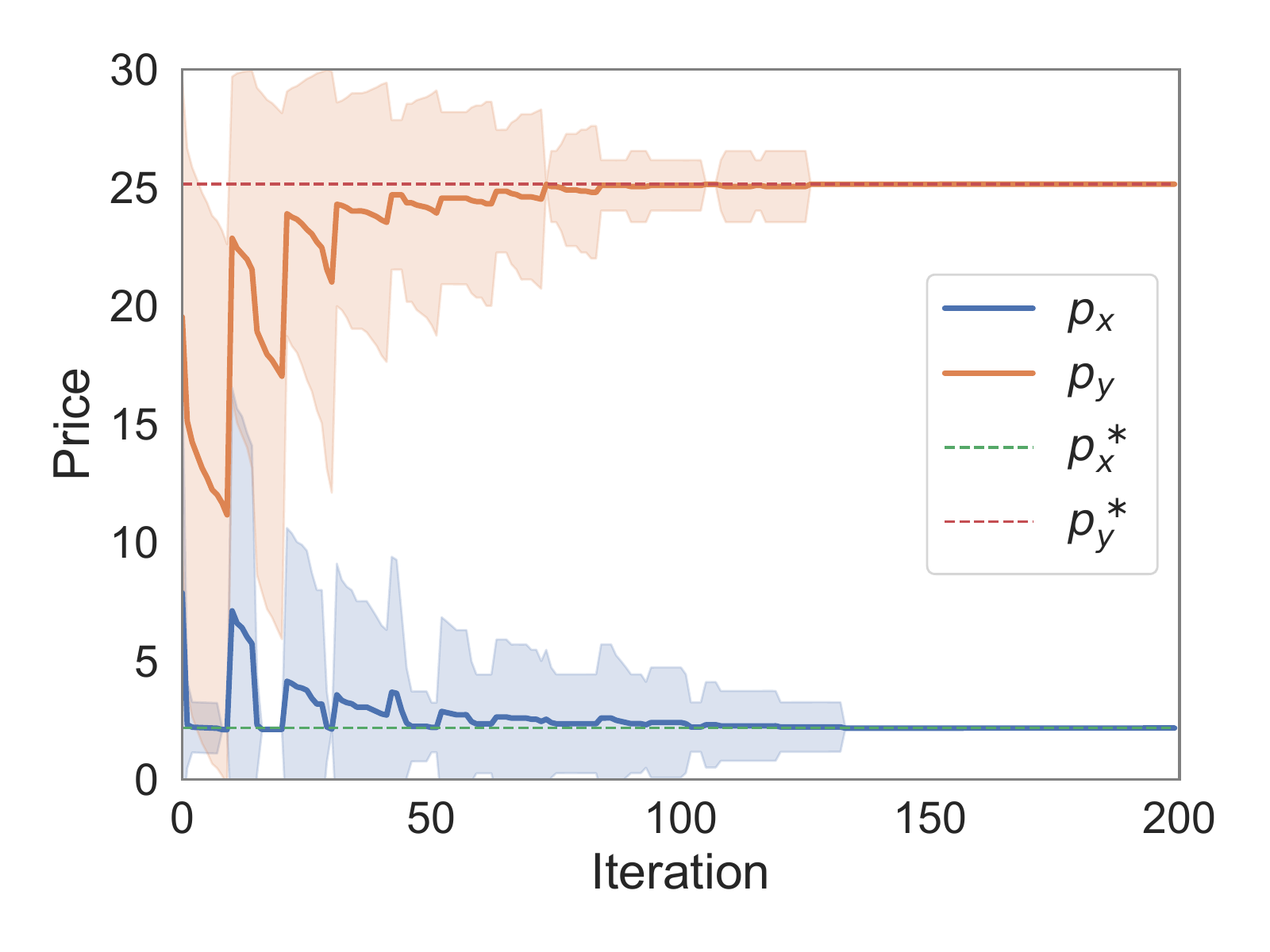}
        \caption{Learning curves for the equilibrium at $(2.168, 25.157)$.}
        \label{fig:eq3}
    \end{subfigure}%
    \quad
    \begin{subfigure}[b]{0.45\linewidth}
        \includegraphics[width=\linewidth]{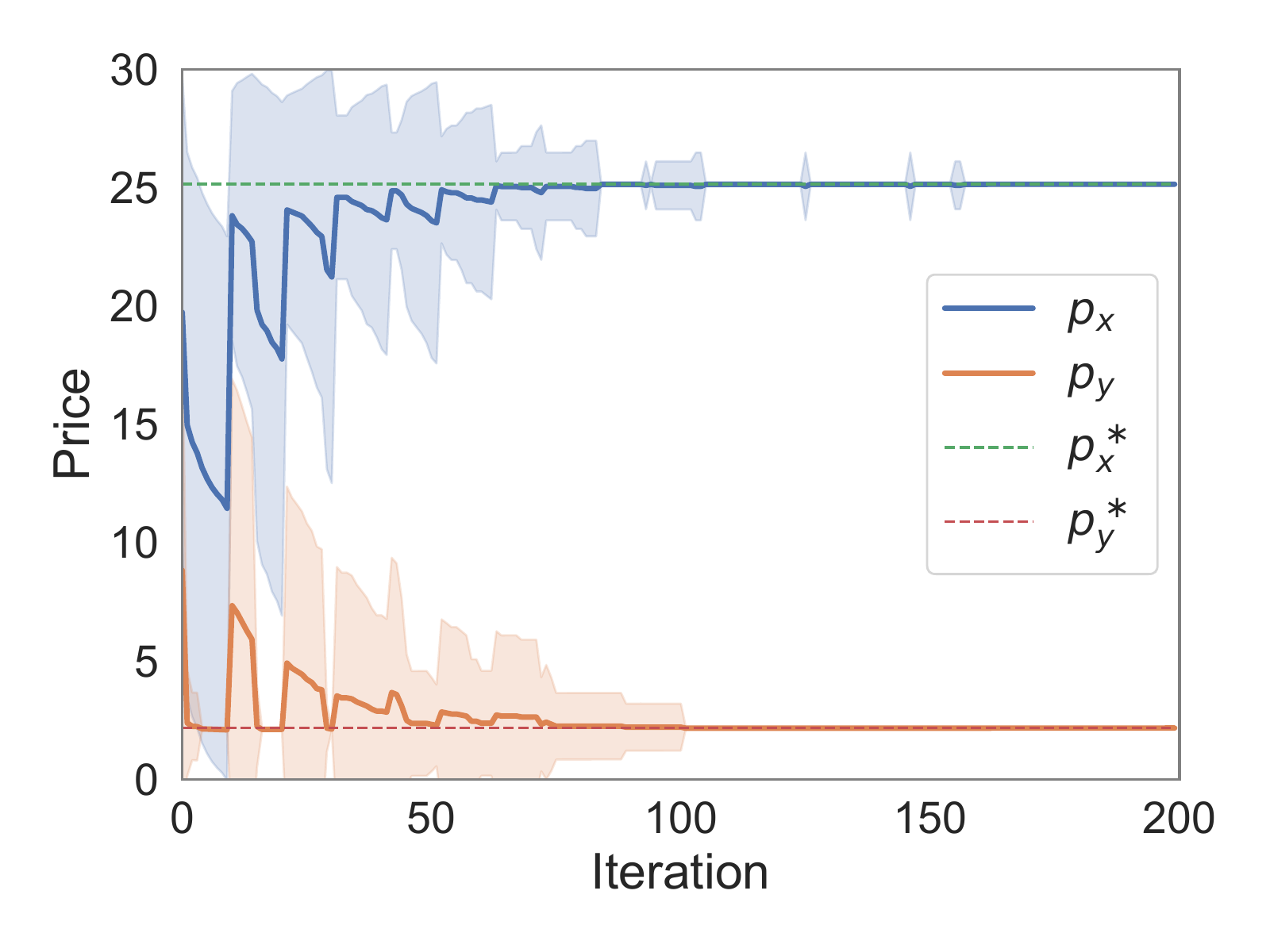}
        \caption{Learning curves for the equilibrium at $(25.157, 2.168)$.}
        \label{fig:eq4}
    \end{subfigure}%
    \caption{Learning curves for the Bertrand price game in the interval $[1, 30]$.}
    \label{fig:bertrand-price-game-full}
\end{figure}

\subsubsection{Unsuitable Approximation}
Next, we consider what happens when an unsuitable approximation of the strategy sets is used. Specifically, we raise the minimum price to $4$, which renders both equilibria from \cref{tab:bertrand-eq} impossible. Intuitively, these equilibria had one firm that opted for a mass-market strategy with lower prices and another that opted for a niche strategy with higher prices. By raising the minimum price, we render this mass-market strategy impossible. We show the resulting learning trajectories for this experiment in \cref{fig:bertrand-price-game-restricted}.

We find that throughout all trials the firms rapidly converge to the joint strategy $(22.987, 22.987)$, which leads to a profit of $673.38$ for both. In fact, this joint strategy is a Nash equilibrium in the MONFG but not a Nash equilibrium of the original continuous game, thus demonstrating the mentioned limitations of the approximation technique. Interestingly, the Nash equilibrium in the approximate MONFG leads to a joint strategy with higher social welfare when considering both the total sum of profits as well as the maximum lowest profit. Specifically, it returns a total profit of $1346.754$, while both equilibria from \cref{tab:bertrand-eq} lead to a total profit of $1333.318$ and has a higher lowest profit. As such, even when using insufficient approximations for the original game, our contributions may result in interesting solutions from, e.g., a mechanism design perspective. 

We note that designing appropriate approximations for arbitrary continuous games is a nontrivial task. For most interesting applications, suitable convex compact approximations will not be given and thus require leveraging domain knowledge or post-processing to confirm the retrieved solution in the original game.

\begin{figure}[]
    \centering
    \includegraphics[width=0.75\linewidth]{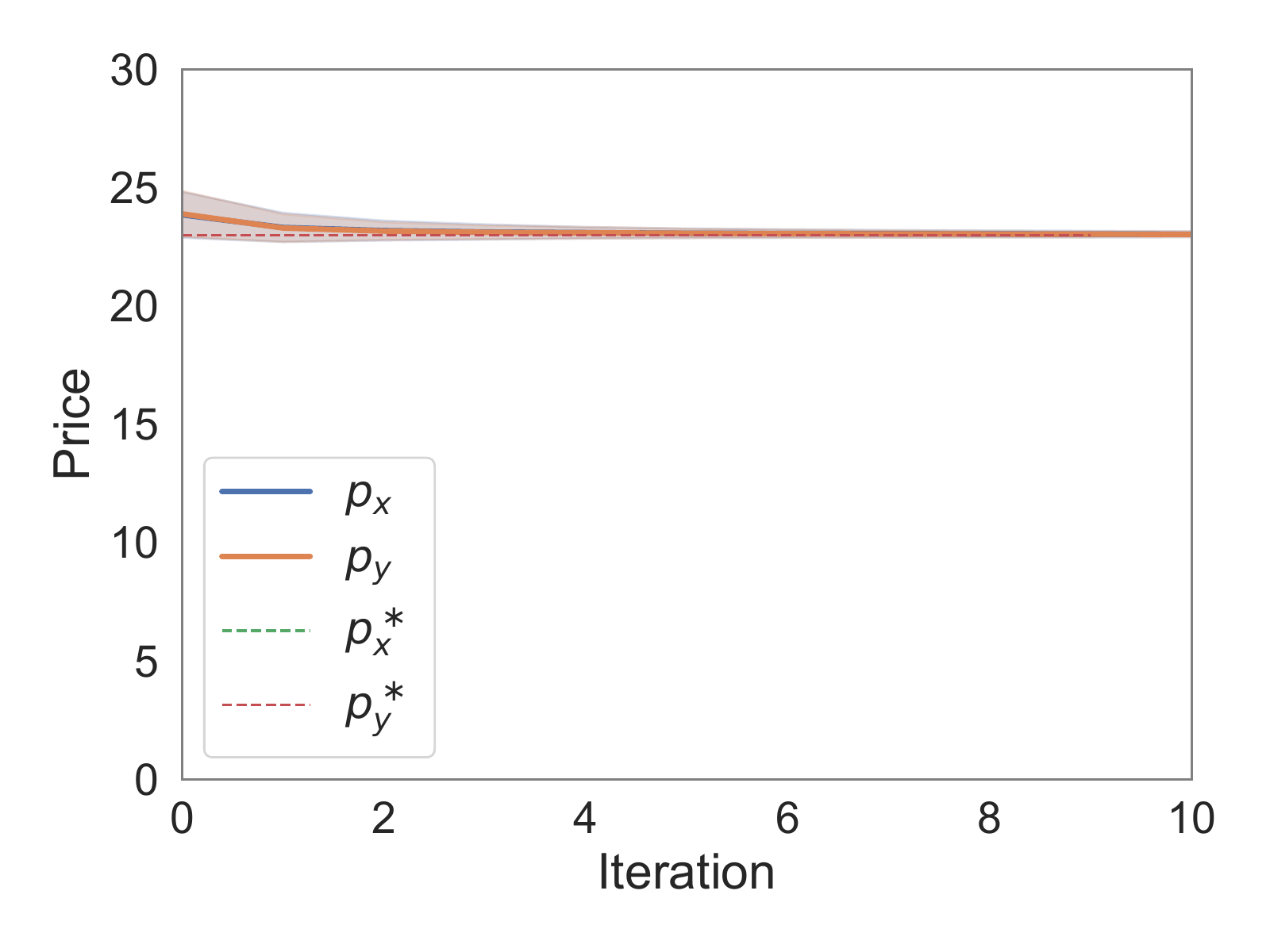}
    \caption{Learning curves for the Bertrand price game in the interval $[4, 30]$.}
    \label{fig:bertrand-price-game-restricted}
\end{figure}

%%%%%%%%%%%%%%%%%%%%%%%%%%%%%%%%%%%%%%%%%%%%%%%%%%%%%%%%%%%%%%%%%%%%%%%%
\section{Related Work}
\label{sec:related-work}
\balance
Multi-objective games were first introduced by Blackwell \citep{blackwell1954analog} and have since been studied broadly. One solution concept that is often considered are Pareto-Nash equilibria, which take a utility agnostic approach and are thus defined over the vectorial payoffs directly \citep{ismaili2018existence}. When taking a utility-based approach \citep{roijers2017multiobjective}, it has been shown that Nash equilibria need not exist \citep{radulescu2020utilitybased}. Follow up work showed that assuming only quasiconcave utility functions \citep{ropke2022nash} is a sufficient condition to guarantee existence again. From a computational perspective, reinforcement learning algorithms and additional techniques such as opponent modelling \citep{zhang2020opponent,radulescu2022opponent} and communication \citep{ropke2022preference} have recently been successfully explored.

The second game model we considered in this work are continuous games. General existence guarantees are known in these games, with for example work on Nash equilibria \citep{glicksberg1952further} and correlated equilibria \citep{hart1989existence}. From a computational perspective, both fictitious play \citep{ganzfried2021algorithm} and no-regret learning have been explored \citep{hsieh2021adaptive}, with the latter also obtaining strong convergence guarantees. We note that while we adhere to the definition of a continuous game by Stein et al. \citep{stein2008separable}, there exist other definitions for continuous games which place different assumptions on the strategy sets or utility functions \citep{ganzfried2021algorithm,hsieh2021adaptive,adam2021double}. Finally, polynomial games, a subset of continuous games, have been covered in detail with theoretical and algorithmic results for their Nash equilibria and correlated equilibria \citep{stein2008separable,stein2011correlated}. 

Our work is related to other equivalence notions in game theory. The first notable example of such an equivalence notion is strategic equivalence \citep{maschler2013game}. An advantage of strategic equivalence is that Nash equilibria are preserved, thus being a useful construct for computing Nash equilibria in games. For example, the Nash equilibria of an otherwise intractable game might be computed by constructing a strategically equivalent zero-sum game for which efficient solving methods do exist \citep{heyman2019computation}. 

Pure strategy equivalence, as defined in \cref{sec:pure-strategy-equivalence}, is most closely related to the concept of a game isomorphism which defines two games to be isomorphic when there exists a mapping from one to the other \citep{gabarro2011complexity}. Two variants of a game isomorphism are defined, namely a strong and a weak isomorphism, with a strong isomorphism preserving all Nash equilibria and a weak isomorphism preserving only the pure strategy Nash equilibria.

%%%%%%%%%%%%%%%%%%%%%%%%%%%%%%%%%%%%%%%%%%%%%%%%%%%%%%%%%%%%%%%%%%%%%%%%

\section{Conclusion}
\label{sec:conclusion}
We contribute a novel equivalence class, called pure strategy equivalence, between continuous games and multi-objective normal-form games. We show that for every continuous game whose strategy spaces are convex subsets of an Euclidean space, a pure strategy equivalent MONFG can be constructed and vice versa. Moreover, this equivalence entails the persistence of Nash equilibria. We demonstrate the applicability of pure strategy equivalence by learning Nash equilibria in two continuous games utilizing a multi-objective fictitious play algorithm.

The equivalence notion introduces a range of new theoretical contributions and computational approaches for both models. For the theoretical aspect, it is known that other game models which allow for more complex interactions, such as Bayesian games and extensive form games, can be reduced to normal-form games \citep{maschler2013game}. Formulating a bridge between continuous games and multi-objective games thus opens the possibility for additional equivalence results between these games with an infinite number of pure strategies and the related multi-objective variant. From an algorithmic perspective, we expect that this will allow continuous games to be solved more efficiently, as the tabular structure of MONFGs can be more appropriate for computational approaches. For future work, we aim to evaluate this on larger continuous games with more complex structures. Finally, as the equivalence notion is not unique, an interesting question remains how to find the best pure strategy equivalent game for any given game.

\begin{acks}
We are grateful to Raphael Avalos for helpful discussions during the early stages of this work. WR and RR are supported by the Research Foundation – Flanders (FWO), grant numbers 1197622N and 1286223N and CG is supported by the Marie-Curie grant 101063180. This research was supported by funding from the Flemish Government under the ``Onderzoeksprogramma Artifici\"{e}le Intelligentie (AI) Vlaanderen'' program.
\end{acks}
%\newpage

%%%%%%%%%%%%%%%%%%%%%%%%%%%%%%%%%%%%%%%%%%%%%%%%%%%%%%%%%%%%%%%%%%%%%%%%

%%% The next two lines define, first, the bibliography style to be 
%%% applied, and, second, the bibliography file to be used.

\bibliographystyle{ACM-Reference-Format} 
\bibliography{bibliography}

%%%%%%%%%%%%%%%%%%%%%%%%%%%%%%%%%%%%%%%%%%%%%%%%%%%%%%%%%%%%%%%%%%%%%%%%
\newpage
\appendix

\section{Measure Theory}
\label{ap:measure-theory}
In this work, and in \cref{sec:mixed-strategy-equivalence,sec:mapping-of-ne} in particular, we make use of results from the mathematical field of measure theory. To make this work relatively self-contained, we provide a short supplementary section on elementary results employed in our definitions and proofs. For a more in-depth treatment of this subject we refer to the textbook of Capi\'{n}ski and Kopp \cite{capinski2004measure}.

\subsection{Measures}
A measure generalises the concepts of length and volume. Informally, a measure is a function which maps subsets of some parent set to a non-negative real number, corresponding to its ``mass''. 

Let $X$ be a set. A $\sigma$-algebra on this set is a nonempty collection $\mathcal{F}$ of subsets of $X$ closed under complement, countable unions and countable intersections. We call the pair $(X, \mathcal{F})$ a measurable space. A function $\mu: \mathcal{F} \to [0, +\infty]$ is called a \emph{measure} if, 

\begin{enumerate}
    \item $\mu(\emptyset) = 0$;
    \item $\mu \left(\bigcup^\infty_{i=1} B_i\right) = \sum_{i=1}^n \mu \left(B_i \right)$ for all $B_i \in \mathcal{F}$ with $B_i \cap B_j = \emptyset, i \neq j$.
\end{enumerate}

Given a measurable space $(X, \mathcal{F})$ and a measure $\mu$ on $(X, \mathcal{F})$, $(X, \mathcal{F}, \mu)$ is called a measure space.

To illustrate this concept, consider a specific example called the Dirac measure. A Dirac measure assigns a value of 1 to every subset which contains a predefined element $x$ and a value of 0 to all others. This measure can be used to define the Dirac delta function, a useful tool in physics. Formally, a Dirac measure $\mu_x$ on the measurable space $(X, \mathcal{F})$ for a given $x \in X$ is defined as,
\begin{equation*}
    \mu_x(B) = \begin{cases}
0, \quad x \notin B\\
1, \quad x \in B
\end{cases}
\quad \text{for all } B \in \mathcal{F} 
\end{equation*}

\subsection{Borel Probability Measures}
Roughly speaking, a topological space is a set together with a collection of subsets that have been declared to be open. A Borel set is any set in a topological space that can be formed from open sets by taking countable unions, countable intersections and relative complements. The set of all Borel sets on $X$ is a $\sigma$-algebra, called the Borel $\sigma$-algebra, and is the smallest $\sigma$-algebra containing all open sets in $X$. As such, given a set $X$ together with a Borel $\sigma$-algebra $\mathcal{F}$, we may define the measure space $(X, \mathcal{F})$. A measure $\mu$ on this space is called a Borel measure.

Let $(X, \mathcal{F})$ be a measure space and $\mu$ a measure on this space. If $\mu(X) = 1$, $\mu$ is called a probability measure. Naturally, if $\mathcal{F}$ is the Borel $\sigma$-algebra, $\mu$ is also known as a Borel probability measure.

In continuous games, the set of actions is a nonempty compact metric space. As such, when considering mixed strategies, one cannot simply assign an individual probability to each action. Rather, we require measures to be defined over the open sets in the action space to assign probabilities. Because the Borel $\sigma$-algebra is the smallest $\sigma$-algebra which contains all open sets in the action space, it is standard to define mixed strategies to be Borel probability measures defined over this space. The set of all mixed strategies is further naturally defined as the set of all Borel probability measures.

\subsection{Pushforward Measure}
A homeomorphism $\varphi:X\to Y$ between topological spaces is a continuous bijection with a continuous inverse; in particular, $\varphi^{-1}(B)$ and $\varphi(A)$ are open whenever $B\subseteq Y,A\subseteq X$ are open.
Given a Borel measure $\mu$ on $X$, the \textit{pushforward measure} $\varphi_*(\mu)$ on $Y$ is defined as
\begin{equation*}
    \varphi_{*}\left(\mu\right)\left(B\right)=\mu\left(\varphi^{-1}\left(B\right)\right)
\end{equation*}
for all sets $B$ in the Borel $\sigma$-algebra $\mathcal{F}_{Y}$ of $Y$. It can be checked that this indeed defines a Borel measure on $Y$.
We need the following folklore result, for which we provide a proof for completeness.
\begin{lemma}
\label{lemma:pushforward-bijective}
Let $X,Y$ be topological spaces with Borel $\sigma$-algebras $\mathcal{F}_X, \mathcal{F}_Y$ and sets of Borel probability measures $\mathcal{B}(X),\mathcal{B}(Y)$ respectively. If $\varphi:X\to Y$ is a homeomorphism, then the map
\[
\psi:\mathcal{B}(X)\to \mathcal{B}(Y): \mu \mapsto \varphi_*(\mu)
\] is a bijection.
\end{lemma}
\begin{proof}
Let us first remark that 
\begin{equation}
\label{eq:gen}
\{\varphi^{-1}(B):B \subseteq Y\text{ open}\}=\{A:A \subseteq X\text{ open}\}
\end{equation}
because $\varphi$ is a homeomorphism. Therefore, the set above generates $\mathcal{F}_X$, and a Borel measure on $X$ is uniquely determined via the values it takes on that set. 

To see that $\psi$ is injective, let $\mu,\mu'\in \mathcal{B}(X)$. Then
\begin{align*}
    & \psi(\mu) = \psi(\mu') \\
    \implies & \varphi_*(\mu) = \varphi_{*}(\mu') \\
    \implies & \mu\left(\varphi^{-1}\left(B\right)\right) = \mu'_i\left(\varphi^{-1}\left(B\right)\right) \forall B\subseteq Y\text{ open} \\
    \implies &  \mu = \mu'.
\end{align*}
The last step uses again that (\ref{eq:gen}) is a generating set for $\mathcal{F}_X$, which implies that any two Borel measures that take the same value on this set, must be equal.

To see that $\psi$ is also surjective, let $\mu_Y \in \mathcal{B}(Y)$ and we simply define $\mu_X\in \mathcal{B}(X)$ on the generating set (\ref{eq:gen}) again:
\begin{equation*}
   \mu_X(\varphi^{-1}(B))= \mu_Y(B), \quad \forall B \subseteq Y \text{ open}.
\end{equation*}
By construction, $\psi(\mu_X)=\mu_Y$ since for all open $B\subseteq Y$,
\[
\psi(\mu_X)(B)=\mu_X(\varphi^{-1}(B))= \mu_Y(B),
\]
so since the measures take the same values on a generating set for $\mathcal{F}_Y$, they take the same values on $\mathcal{F}_Y$.
\end{proof}

\section{Identity Game}
\label{ap:identity-game}
In this section, we provide a formal proof for \cref{lemma:identity-game}. We restate the lemma below. 

\begin{lemma1}[Identity Game]
For any finite set of players and finite sets of pure strategies, there exists a set of payoff functions $\bm{p}$ such that for each player $i$, $\bm{p}_i(\delta) = \delta$. 
\end{lemma1}
%This lemma states that for every finite set of players with finite sets of pure strategies, a set of payoff functions $\bm{p}$ exists such that for each player $i$, $\bm{p}_i(\delta) = \delta$.

\begin{proof}
We show that an identity game can always be constructed for a finite set of players $N$ with a finite set of pure strategies $A_i \in \mathcal{A}$ for each player $i \in N$. 

We first specify the description of strategies $s_i$ in the identity game as it determines the required payoff. Let $m_i \coloneqq |A_i|$. Each player's strategy is a probability vector of length $|A_i|$,
\begin{equation}
    \forall i \in N: \delta_{i} = \left(P\left(a_{i, 1} | \delta_{i}\right), \dots, P\left(a_{i, m_i} | \delta_{i} \right) \right).
\end{equation}

Let $s$ be a joint strategy with an individual strategy $\delta_i$ for each player $i$. The description of $s$ is the concatenation of all individual strategies $\delta_i$ in ascending order. In games where each player has a finite number of pure strategies, $\delta$ must necessarily have a finite description length. Specifically, because each individual strategy has length $m_i$ and the joint strategy $\delta$ combines all such strategies in one vector, we find that
\begin{equation}
    |\delta| = \sum_{i = 1}^{n} m_i.
\end{equation}

We denote the probability of player $i$ playing action $a_{i,l}$ in in the joint-strategy $\delta$ as $\delta_{i,l}$ so that $\delta_{i,l} = P(a_{i,l} | \delta)$.

An identity game ensures that for all players $i$, $\bm{p}_i(\delta) = \delta$. As such, we need only define one payoff function, which is shared by all players. Furthermore, observe that the length of each payoff vector is equal to the length of the joint strategy,
\begin{equation}
    \forall i \in N, \forall a \in \mathcal{A}: \bm{p}_i(a) \in \mathbb{R}^{|\delta|}.
\end{equation}

A pure joint strategy $a$ is the special case of a joint strategy $\delta$ where each player $i$ deterministically plays an action $a_i \in A_i$. We denote this by $\delta_a$ whenever we want to specify the full joint strategy notation. As a last remark on notation, we define $\delta_i(a_i)$ to be the probability of playing $a_i \in A_i$ under $\delta_i$.

We now construct the payoffs for the identity game. Let the finite set of payoffs for pure strategies equal the joint strategy vector, i.e. 

\begin{align}
\label{eq:identity-eq}
\bm{p}_i(a) & = \delta_a \\
    & = \left(P \left( a_{1,1} | \delta_a \right), \cdots , P \left( a_{n, {m_n}} | \delta_a \right) \right).
\end{align}
We show that this ensures $\bm{p}_i(\delta) = \delta$. The expected payoff vector of a mixed joint strategy $\delta$ is calculated as follows (see Eq. \ref{eq:exp-vec-payoff}),
\begin{align*}
    \bm{p}_i(\delta) & = \sum_{a \in \mathcal{A}}\bm{p}_i(a)\prod_{j = 1}^n \delta_j(a_j).
\end{align*}
Furthermore, because strategies are independent,
\begin{equation*}
    P(a | \delta) = \prod_{j = 1}^n \delta_j(a_j).
\end{equation*}
By the law of total probability we can thus say that,
\begin{align}
    \bm{p}_i(\delta) & = \sum_{a \in \mathcal{A}}\bm{p}_i(a)\prod_{j = 1}^n \delta_j(a_j) \\
    & = \sum_{a \in \mathcal{A}} \left(P \left( a_{1,1} | a \right), \cdots , P \left( a_{n, m_n} | a \right) \right) P(a | \delta) \\
    & = \sum_{a \in \mathcal{A}} \left(P \left( a_{1,1} | a \cap \delta \right), \cdots , P \left( a_{n, m_n} | a \cap \delta \right) \right)P(a | \delta) \label{eq:conditional-prob} \\
    & = \sum_{a \in \mathcal{A}} \left(P \left( a_{1,1} | a \cap \delta  \right) P(a | \delta), \cdots , P \left( a_{n, m_n} | a \cap \delta  \right) P(a | \delta) \right) \\
    & = \left(P \left( a_{1,1} | \delta \right), \cdots , P \left( a_{n, m_n} | \delta \right) \right) \\
    & = \delta.
\end{align}

Note that \cref{eq:conditional-prob} holds because the probability of playing any action $a_{i,l}$ is uniquely defined under the joint action $a$ and is independent from another joint strategy $\delta$.
\end{proof}

\section{Pure Strategy Equivalence}
\label{ap:pure-strategy-equivalence}
We provide explicit construction methods that transform a game from one class to the other and demonstrate that such constructions are not necessarily unique. In addition, we present a standard approach for the construction of the strategy bijections between the continuous game and MONFG. Finally, we briefly discuss the potential pitfalls of this approach and how future work could address these challenges.

\subsection{Constructing Equivalent Games}
\label{ap:constructing-games}
The proofs presented in \cref{sec:pure-strategy-equivalence} can be formalised to compute pure strategy equivalent games. We first show the construction of a pure strategy equivalent continuous game starting from an MONFG in \cref{alg:cg-construction}. Observe that because both the player bijection and strategy bijection are identity functions, the construction can be performed efficiently.

\begin{algorithm}[th]
\caption{The continuous game construction from an MONFG.}
\label{alg:cg-construction}
    \begin{algorithmic}[1]
    \Require An MONFG $G_m = (N_m, \mathcal{A}, \bm{p})$ and utility functions $u$
    \Ensure A continuous game $G_c = (N_c, \mathcal{S}, v)$
    \State $N_c \gets N_m$
    \State $\mathcal{S} \gets \Delta^{k_1} \times \dots \times \Delta^{k_n}$
    \State $v \gets (u_{1} \circ \bm{p}_1, \dots, u_{n} \circ \bm{p}_n)$
    \State $G_c \gets (N_c, \mathcal{S}, v)$
    \State \Return $G_c$
    \end{algorithmic}
\end{algorithm}

The construction from a continuous game to an MONFG is shown in \cref{alg:monfg-construction}. Here, we explicitly require strategy bijections $\varphi_i$ to be given which map strategy sets to simplices and are used in constructing the utility functions. The efficiency of computing the utility for a given strategy in the MONFG is then dependent on the efficiency of computing the inverse of the strategy bijection. 

\begin{algorithm}[th]
\caption{The MONFG construction from a continuous game.}
\label{alg:monfg-construction}
    \begin{algorithmic}[1]
    \Require A continuous game $G_c = (N_c, \mathcal{S}, v)$ and homeomorphisms $\varphi_i: S_{i} \to \Delta^{k_i}$
    \Ensure An MONFG $G_m = (N_m, \mathcal{A}, \bm{p})$ and utility functions $u$
    \State $N_m \gets N_c$
    \State $\mathcal{A} \gets [k_1 + 1] \times \dots \times [k_n + 1]$
    \State $\bm{p} \gets $ \textproc{IdentityPayoffs}$(N_m, \mathcal{A})$
    \State $G_m \gets (N_m, \mathcal{A}, \bm{p})$
    \State $u \gets (v_{1} \circ \varphi^{-1}_1, \dots, v_{n} \circ \varphi^{-1}_n)$
    \State \Return $G_m, u$
    \end{algorithmic}
\end{algorithm}

Note that $[k_i + 1] =\{1,\dots,k_i + 1\}$ where $k_i$ is the number of vertices in the simplex homeomorphic to player $i$'s strategy set. The function \textproc{IdentityPayoffs} returns the payoffs of the identity game for a given player base and joint action set.

To conclude this section, we provide a proof for \cref{re:non-uniqueness}. Informally, we noted that one can construct MONFGs with different payoff and utility functions which are all pure strategy equivalent to a given continuous game and vice versa. While the games still share the same structure, having different payoff and utility functions may ensure that some games are computationally simpler to solve than others. Identifying games that are strategically equivalent to efficiently solvable games has been studied with success in other settings as well \cite{heyman2019computation}.

\begin{remark1}
A continuous game may have multiple pure strategy equivalent multi-objective normal-form games and vice versa.
\end{remark1}

\begin{figure}[h]
    \centering
    \begin{game}{2}{2}
                  & $A$        & $B$\\
        $A$   & $(0, 0); (0, 0)$       & $(0, 3); (0, 3)$ \\
        $B$   & $(3, 0); (3, 0)$       & $(3, 3); (3, 3)$
    \end{game}
    \caption{The MONFG used in the proof of \cref{re:non-uniqueness}.}
    \label{fig:pse-monfg}
\end{figure}

\begin{proof}
Let $G_{c}$ be a continuous game with $S_{1} = S_{2} = [0, 3]$ and utility functions $v_{1}(s_1, s_2) = v_{2}(s_1, s_2) = s_1 + s_2$. Observe that the strategy spaces are already simplices, but are not probability simplices. To construct an MONFG, we can use \cref{alg:monfg-construction} and strategy bijection $\varphi_i(s) = \frac{s}{3}$ for both players. The resulting pure strategy equivalent MONFG, called $G_m$, has the same payoffs as shown in \cref{fig:identity-game-2-2} and utility functions $u_i(p_1, p_2, p_3, p_4) = 3 \cdot p_2 + 3 \cdot p_4$. Consider now a second MONFG, $G'_m$, shown in \cref{fig:pse-monfg} with the utility function $u_i(p1, p2) = p1 + p2$ for both players. It is clear that this game is also pure strategy equivalent to $G_c$, as it directly represents the simplex strategy spaces. As such, $G_c$ is pure strategy equivalent to both $G_m$ and $G'_m$.

For the opposite direction, we can apply \cref{alg:cg-construction} to the game shown in \cref{fig:pse-monfg} to obtain a continuous game $G'_c$ with strategy spaces $S_{i} = [0, 1]$ and utility functions $v_i = u_i \circ \bm{p}_i$ for both players. As such, $G'_m$ is pure strategy equivalent to both $G_c$ and $G'_c$.
\end{proof}

\subsection{Constructing the Strategy Bijections}
\label{ap:constructing-homeomorphisms}
To construct an MONFG from a continuous game using \cref{alg:monfg-construction}, it is necessary to provide a strategy bijection $\varphi_i: S_i \to \Delta^{k_i}$ for every player $i$. We provide a straightforward, albeit likely inefficient, approach for obtaining such functions. Afterwards, we present a short discussion on other techniques that may be better suited for this task.

\subsubsection{Standard Technique}
When no obvious homeomorphism is available, it is possible to first construct a map from any player $i$'s strategy space $S_i$ to the closed unit ball $B \subset \mathbb{R}^{k_i}$. We can subsequently create a homeomorphism from $B \subset \mathbb{R}^{k_i}$ to a probability simplex $\Delta^{k_1}$. By composing the two, we obtain a homeomorphism from $S_i$ to the probability simplex $\Delta^{k_1}$. 

Let $C$ be a compact convex subset in $\mathbb{R}^d$ with nonempty interior and define $\partial C$ as its boundary. Let $f: \partial C \to S^{d-1}$ be defined by
\begin{equation}
f(x) = \frac{x}{\|x\|}.
\end{equation}
This maps the boundary points of $C$ to the $(d-1)$-sphere. Intuitively, for a 2-dimensional convex compact set $C$, $f$ maps the boundary $\partial C$ to a circle. The map $f$ is a homeomorphism, so has an inverse $f^{-1}$, and the map
$k: B^d \to C$
defined by
\begin{equation}
\label{eq:homeomorphism}
k(x) = 
\begin{cases}
\|x\| f^{-1}\left(\frac{x}{\|x\|}\right) & \quad x \neq 0, \\
0 & \quad x = 0,
\end{cases}
\end{equation}
is also a homeomorphism (see e.g. \cite{bredon1993general}). Note that the construction assumes the origin is in the interior of $C$, which is always possible to accomplish by translation.

To complete the full construction, we also specify the inverse functions $f^{-1}: S^{d-1} \to \partial C$ and $k^{-1}: C \to B^d$. First, $x=f^{-1}(y)$ can be obtained by noticing that $x$ is the place where the ray through $y$ intersects $\partial C$. Let $p_{A}(y) := \inf\{\lambda > 0: y \in \lambda A\}$ be a Minkowski functional. Informally, a Minkowski functional $p_A$ returns for an input point $y$ the smallest positive number by which it is possible to scale $A$ such that $y$ is in the resulting space. We define
\begin{equation}
\label{eq:inverse-f}
f^{-1}(y) = \frac{y}{p_{\partial C}(y)}.
\end{equation}
Finally, $k^{-1}:C\to B^d$ can be constructed by first computing where the ray through $y$ intersects $\partial C$ and rescaling:
\begin{equation}
\label{eq:inverse-k}
k^{-1}(y) = 
\begin{cases}
    f\left(\frac{y}{\|y\|}\right)\frac{1}{\|y\|} & \quad y \neq 0, \\
    0 & \quad y = 0,
\end{cases}.
\end{equation}

As stated earlier, we may use these functions to go from any compact convex subset of an Euclidean space $C$ to the closed unit ball and we can apply the same procedure to map from the probability simplex to the unit ball. By composing the two, a full homeomorphism is obtained between $C$ and the probability simplex.

% It is possible to show that the function defined in \cref{eq:homeomorphism} is a homeomorphism from a compact convex subset $C$ in $\mathbb{R}^d$ with nonempty interior to the closed unit ball $B$. 
% \begin{equation}
% \label{eq:homeomorphism}
% s(\bm{z}) := \begin{cases}
% \frac{p_C(\bm{z})}{\|\bm{z}\|} \bm{z}       & \quad \bm{z} \in C \setminus \{\bm{0}\}\\
%     \bm{0}   & \quad \bm{z} = \bm{0}
% \end{cases}
% \end{equation}
% Note that this function uses a Minkowski functional $p_C$. Let $A \subseteq X$ be a subset of a real vector space. Then, a Minkowski functional $p_C$ is defined as follows,
% \begin{equation}
% \label{eq:minkowski-functional}
% p_{A}(\bm{x}) := \inf\{\lambda \geq 0: \bm{x} \in \lambda A\}.
% \end{equation}

\subsubsection{Discussion}
While the proposed approach is straightforward to explain, it may be difficult to implement in practice. This is because computing the Minkowski functional $p_{\partial C}$ requires searching over a continuous range. One possible solution is to utilize a binary search algorithm that locates a $\lambda$ with a desired precision. However, as this approach may be computationally expensive, we suggest exploring more efficient techniques whenever possible. Another option is to cache values of $\lambda$ and reuse them when feasible to avoid the need for repeated binary searches.

When leveraging pure strategy equivalence to solve continuous games, it may be beneficial to employ algorithms that necessitate only a limited number of function evaluations to avoid costly computations. Finally, rather than employing an exact solution for the strategy bijections, it may be useful to \emph{learn} such bijections. For future work, we aim to study this in more detail.

\section{Mixed Strategy Equivalence}
\label{ap:mixed-strategy-equivalence}
In this section, we first provide a proof for \cref{th:cg-mse-monfg}. We reiterate the theorem here for completeness.

\begin{theorem4}
If a continuous game is pure strategy equivalent to a multi-objective normal-form game, they are also mixed strategy equivalent.
\end{theorem4}

%Concretely, we show that pure strategy equivalence as defined in Definition 3.1 implies mixed strategy equivalence as defined in Definition 3.2.

\begin{proof}
Let $G_c = (N_c, \mathcal{S}, v)$ be a continuous game and $G_m = (N_m, \mathcal{A}, \bm{p})$ be a pure strategy equivalent multi-objective normal-form game with utility functions $u$. For notational simplicity, we assume the player bijection $\pi: N_c \to N_m$ to be implicitly applied in any mapping between $G_c$ and $G_m$ and refer to the set of players simply as $N$.

Let $i\in N$. The set $\mathcal{B}(S_{i})$ is the set of Borel probability measures on the strategy set of player $i$ and thus represents their mixed strategies; similarly, $\mathcal{B}(\Delta^{k_i})$ denotes the Borel probability measures on $\Delta^{k_i}$, with $k_i$ from the definition of pure strategy equivalence.
As $G_c$ and $G_m$ are pure strategy equivalent, there exists a continuous bijective function $\varphi_i$ with a continuous inverse mapping pure strategies from $G_c$ to mixed strategies in $G_m$ for each player $i$. 
By \cref{lemma:pushforward-bijective}, the function
\begin{equation*}
\psi_i: \mathcal{B}(S_{i}) \to \mathcal{B}(\Delta^{k_i})
\end{equation*}
with $\psi_i(\mu_{i}) = \varphi_{i*}(\mu_{i})$ is a bijection. Let $\psi = \psi_1 \times \dots \times \psi_n$, which is then also a bijection.

It remains to show that the utility from a mixed strategy $\mu$ in $G_c$ equals that of its corresponding hierarchical strategy $\psi(\mu)$ in $G_m$.
\begin{align*}
    \mathbb{E} v_{i}(\mu) & = \int_{S_1 \times \cdots \times S_n} v_{i}(s_1, \cdots , s_n) d\mu_{1}(s_1) \cdots d\mu_{n}(s_n) \\
    % & = \int_{S'_1 \times \cdots \times S'} u_{c,i}\left(\varphi^{-1}_i\left(s'_1, \cdots , s'_n\right)\right) d(\varphi^{-1}_i(\tilde{s}_1))(\varphi^{-1}_i(s'_1)) \cdots d(\varphi^{-1}_i(\tilde{s}_n))(\varphi^{-1}_i(s'_n)) \\
    & = \int_{S_1 \times \cdots \times S_n} u_{i}\left(\mathbb{E} \bm{p}_i \left(\varphi_i\left(s_1, \cdots , s_n\right)\right)\right) d\mu_{1}(s_1) \cdots d\mu_{n}(s_n) \\
    & = \int_{\Delta^{k_1} \times \cdots \times \Delta^{k_n}} u_{i}\left(\mathbb{E} \bm{p}_i \left(\delta_1, \cdots , \delta_n\right)\right) d\psi_1(\mu_1)(\delta_1) \cdots d\psi_n(\mu_n)(\delta_n) \\
    & = \mathbb{E} u_{i}\left(\psi\left(\mu\right)\right). \qedhere
\end{align*}
\end{proof}

\section{Mapping of Nash Equilibria}
\label{ap:mapping-ne}
A major advantage of pure strategy equivalence is that game dynamics remain intact. We formalised this in \cref{th:ps-ms-ne-equivalence} stating that a pure strategy is an NE in a continuous game if and only if it is a mixed strategy NE in a pure strategy equivalent MONFG. Below, we show this formally.

\begin{theorem5}
A pure strategy is a Nash equilibrium in a continuous game if and only if it is a mixed strategy Nash equilibrium in a pure strategy equivalent multi-objective normal-form game.
\end{theorem5}

\begin{proof}
Let $G_c = (N_c, \mathcal{S}, v)$ be a continuous game and $G_m = (N_m, \mathcal{A}, \bm{p})$ be a pure strategy equivalent multi-objective normal-form game with utility functions $u$. For notational simplicity we assume the player bijection $\pi: N_c \to N_m$ to be implicitly applied in any mapping between $G_c$ and $G_m$ and refer to the set of players simply as $N$.

Assume first that $s^\ast$ is a pure strategy Nash equilibrium in $G_c$. Let $\varphi_i$ map a pure strategy for player $i$ in $G_c$ to their mixed strategy in $G_m$. From \cref{def:ps-equivalence} we know that pure strategy equivalence between $G_c$ and $G_m$ ensures that
\begin{equation}
\forall s \in \mathcal{S}: v_{i}(s) = u_{i}\left(\bm{p}_{i} \left(\varphi \left(s\right)\right)\right).
\end{equation}
Furthermore, because $s^\ast$ is a Nash equilibrium, we know that,
\begin{equation}
\forall i \in N, \forall s_i \in S_i: v_{i}\left(s^\ast_i, s^\ast_{-i}\right) \geq v_{i}\left(s_i, s^\ast_{-i}\right).
\end{equation}
Putting the pieces together we get $\forall i \in N, \forall s_i \in S_i$,
\begin{align}
& v_i(s^\ast_i, s^\ast_{-i}) \geq v_i(s_i, s^\ast_{-i}) \\
\implies & u_{i}\left(\bm{p}_{i} \left(\varphi\left(s^\ast_i, s^\ast_{-i}\right)\right)\right) \geq u_{i}\left(\bm{p}_{i} \left(\varphi\left(s_i, s^\ast_{-i}\right)\right)\right)
\end{align}
As such, the image of $s^\ast$ is a Nash equilibrium in $G_m$. 

Assume now that $\delta^\ast$ is a mixed strategy Nash equilibrium in $G_m$. Analogous to before, we know that there exists a bijective function $\varphi_i$ which maps pure strategies from $G_c$ to mixed strategies in $G_m$ and that
\begin{equation}
\forall \delta \in \Delta: u_{i}\left(\bm{p}_{i} \left(\delta\right)\right) = v_{i}\left(\varphi^{-1}\left(\delta\right)\right).
\end{equation}
As $\delta^\ast$ is a Nash equilibrium, we know that,
\begin{equation}
\forall i \in N, \forall \delta_i \in \Delta^{k_i}: u_{i}\left(\bm{p}_{i} \left(\delta^\ast_i, \delta^\ast_{-i}\right)\right) \geq u_{i}\left( \bm{p}_{i} \left(\delta_i, \delta^\ast_{-i}\right)\right).
\end{equation}
Putting the pieces together again it follows that $\forall i \in N, \forall \delta_i \in \Delta^{k_i}$, 
\begin{align}
& u_{i}\left(\bm{p}_{i} \left(\delta^\ast_i, \delta^\ast_{-i}\right)\right) \geq u_{i}\left(\bm{p}_{i} \left(\delta_i, \delta^\ast_{-i}\right)\right) \\
\implies & v_{i}\left(\varphi^{-1} \left(\delta^\ast_i, \delta^\ast_{-i}\right)\right) \geq v_{i}\left(\varphi^{-1} \left(\delta_i, \delta^\ast_{-i} \right) \right).
\end{align}
Note that the last inequality states that no player in $G_c$ has a pure strategy deviation that will increase their utility. However, due to the linearity of expectation this implies that no mixed strategy deviation can increase their utility either. As such, the image of $\delta^\ast$ is a pure strategy Nash equilibrium in $G_c$.
\end{proof}

Finally, we provide a formal proof for \cref{th:ms-hs-ne-equivalence}. This theorem states that the prior result extends to mixed strategy NE in continuous games and hierarchical NE in pure strategy equivalent MONFGs.

\begin{theorem6}
A mixed strategy is a Nash equilibrium in a continuous game if and only if it is a hierarchical Nash equilibrium in a pure strategy equivalent multi-objective normal-form game.
\end{theorem6}

\begin{proof}
Let $G_c = (N_c, \mathcal{S}, v)$ be a continuous game and $G_m = (N_m, \mathcal{A}, \bm{p})$ be a pure strategy equivalent multi-objective normal-form game with utility functions $u$. For notational simplicity we assume the player bijection $\pi: N_c \to N_m$ to be implicitly applied in any mapping between $G_c$ and $G_m$ and refer to the set of players simply as $N$.

Assume first that $\mu^\ast$ is a Nash equilibrium in $G_c$. From \cref{th:cg-mse-monfg}, we know that there exists a function $\psi_i$ for each player $i$ that bijectively maps a mixed strategy in $G_c$ to a hierarchical strategy in $G_m$. This ensures that,
\begin{equation}
\forall \mu \in \mathcal{B} \left(\mathcal{S} \right): v_{i}\left(\mu\right) = u_{i} \left(\psi\left(\mu\right)\right).
\end{equation}
Furthermore, because $\mu^\ast$ is a mixed strategy Nash equilibrium,
\begin{equation}
\forall i \in N, \forall \mu_i \in \mathcal{B}\left(S_i\right): v_{i}\left(\mu^\ast_i, \mu^\ast_{-i}\right) \geq v_{i}\left(\mu_i, \mu^\ast_{-i}\right).
\end{equation}
Combining the two we get $\forall i \in N, \forall \mu_i \in \mathcal{B}(S_i)$ 
\begin{align}
& v_{i} \left(\mu^\ast_i, \mu^\ast_{-i}\right) \geq v_{i} \left(\mu_i, \mu^\ast_{-i}\right) \\
\implies & u_{i} \left(\psi \left(\mu^\ast_i, \mu^\ast_{-i} \right) \right) \geq u_{i} \left(\psi \left(\mu_i, \mu^\ast_{-i} \right) \right)
\end{align}
As such, the image of $\mu^\ast$ is a hierarchical Nash equilibrium in $G_m$. 

Assume now that $\mu^\ast$ is a hierarchical Nash equilibrium in $G_m$. We know that,
\begin{equation}
\forall \mu \in \mathcal{B}(\Delta): u_{i}(\mu) = v_{i}(\psi^{-1}(\mu)).
\end{equation}
From the definition of a hierarchical Nash equilibrium, we can also state for $\mu^\ast$ that,
\begin{equation}
\forall i \in N, \forall \mu_i \in \mathcal{B}(\Delta^{k_i}): u_{i}\left(\mu^\ast_i, \mu^\ast_{-i}\right) \geq u_{i}\left(\mu_i, \mu^\ast_{-i} \right).
\end{equation}
Finally, we get $\forall i \in N, \forall \mu_i \in \mathcal{B}(\Delta^{k_i})$,
\begin{align}
& u_{i}\left(\mu^\ast_i, \mu^\ast_{-i} \right) \geq u_{i} \left(\mu_i, \mu^\ast_{-i} \right) \\
\implies & v_{i} \left(\psi^{-1} \left(\mu^\ast_i, \mu^\ast_{-i} \right) \right) \geq v_{i} \left(\psi^{-1} \left(\mu_i, \mu^\ast_{-i} \right)\right)
\end{align}
As such, the image of $\mu^\ast$ is a Nash equilibrium in $G_c$.
\end{proof}

%%%%%%%%%%%%%%%%%%%%%%%%%%%%%%%%%%%%%%%%%%%%%%%%%%%%%%%%%%%%%%%%%%%%%%%%

\end{document}